\documentclass[12pt, draftclsnofoot, onecolumn]{IEEEtran}
\ifCLASSINFOpdf

\else

\fi
\usepackage{amsmath}
\usepackage{makeidx}  
\usepackage{algorithm}
\usepackage{algorithmic}
\usepackage{graphicx}
\usepackage{subfigure}
\usepackage{epstopdf}
\usepackage{bm}
\usepackage{bbding}
\usepackage{cite}
\usepackage{stfloats}

\usepackage{amssymb}
\usepackage{graphicx}

\usepackage{url}
\newtheorem{proposition}{Proposition}
\usepackage{tabularx,booktabs}
\newcolumntype{C}{>{\centering\arraybackslash}X} 
\usepackage{lipsum}

\usepackage{makecell} 

\usepackage{graphicx}
\usepackage{color}
\newtheorem{thm}{Theorem}
\newtheorem{rem}{Remark}
\newtheorem{lem}{Lemma}
\newtheorem{proof}{proof}

\begin{document}
\linespread{1.38}

\title{IRS Aided MEC Systems with Binary Offloading: A Unified Framework for Dynamic IRS Beamforming}

\author{Guangji Chen, and Qingqing Wu
   \thanks{G. Chen and Q. Wu are with the State Key Laboratory of Internet of Things for Smart City, University of Macau, Macao 999078, China (email: guangjichen@um.edu.mo; qingqingwu@um.edu.mo).}
         }


\maketitle

\begin{abstract}
In this paper, we develop a unified dynamic  intelligent reflecting surface (IRS) beamforming framework to boost the sum computation rate of an IRS-aided mobile edge computing (MEC) system, where each device follows a binary offloading policy. Specifically, the task of each device has to be either executed locally or offloaded to MEC servers as a whole with the aid of given number of IRS beamforming vectors available. By flexibly controlling the number of IRS reconfiguration times, the system can achieve a balance between the performance and associated signalling overhead. We aim to maximize the sum computation rate by jointly optimizing the computational mode selection for each device, offloading time allocation, and IRS beamforming vectors across time. Since the resulting optimization problem is non-convex and NP-hard, there are generally no standard methods to solve it optimally. To tackle this problem, we first propose a penalty-based successive convex approximation algorithm, where all the associated variables in the inner-layer iterations are optimized simultaneously and the obtained solution is guaranteed to be locally optimal. Then, we further derive the offloading activation condition for each device by deeply exploiting the intrinsic structure of the original optimization problem. According to the offloading activation condition, a low-complexity algorithm based on the successive refinement method is proposed to obtain high-quality solutions, which is more appealing for practical systems with a large number of devices and IRS elements. Moreover, the optimal condition for the proposed low-complexity algorithm is revealed. Numerical results demonstrate the effectiveness of our proposed algorithms. Besides, the results illustrate the practical significance of the IRS in MEC systems for achieving coverage extension and supporting multiple energy-limited devices for task offloading, and also unveil the fundamental performance-cost tradeoff of the proposed dynamic IRS beamforming framework.

\end{abstract}



\IEEEpeerreviewmaketitle

\section{Introduction}
With the rapid advancement of Internet of Things (IoT) in recent years, a variety of computation-intensive and latency-sensitive applications, such as autonomous driving, augmented reality, virtual reality, and unmanned aerial vehicles, e.g., \cite{wu2018capacity, chen2019learning, meng2022uav} have emerged to provide real-time machine-to-human and machine-to-machine interactions. The success of these new applications rely on the next generation of wireless networks to accommodate massive IoT devices for real-time computation, communication, and control. Due to the production cost consideration and stringent device size constraint, an IoT device is generally equipped with a capacity-limited battery and a low-performance processor. Thus, how to enhance computational capability for IoT devices to handle intensive computation loads with stringent latency requirements becomes a critical challenge in future IoT networks\cite{cao2018joint}. Conventionally, the cloud computing technique is essential to be exploited for providing IoT devices with abundant computational resources \cite{barbarossa2014communicating}. However, cloud computing may incur long computational latency since cloud servers are generally placed in the remote network core.

To resolve the aforementioned issue, mobile edge computing (MEC) has been proposed as an alternative technology, which bridges the gap between the IoT devices and the cloud by deploying servers at the network edge (e.g., cellular base stations (BSs) or WiFi access points (APs)) \cite{mach2017mobile}. With the help of MEC, IoT devices are allowed to directly offload their data set of a task to and then receive the computational results from the MEC server, which significantly reduces the computational latency compared to cloud computing. Regarding the practical implementation of MEC, computational tasks of devices may be classified into different categories, which depends on the tasks' dependence and partitionability. Accordingly, there are two typical computation-offloading modes in MEC systems, namely partial offloading and binary offloading \cite{8016573}. Specifically, partial offloading allows a task to be flexibly partitioned into two parts with one offloaded to edge servers and the other executed locally. In contrast, for the binary offloading mode, a task is required to be either executed at the local device or offloaded to edge servers. Compared to partial offloading, binary offloading is easier to be implemented and can be applicable for scenarios where tasks  are not partitionable. For both the two typical computation-offloading modes, the joint optimization of computation and communication resource allocations have attracted a lot of research interests in terms of energy consumption minimization \cite{9140412, 8537962}, computation latency minimization \cite{9179779, zhu2020resource}, computation rate maximization \cite{8304010, 8334188, 8434285}, and energy efficiency maximization \cite{8986845, 9312671}. However, the efficiency of task offloading may be severely degraded by the wireless channel attenuation between the AP and devices, which thus fundamentally locks the full potential of MEC. Aiming to tackle this issue, the massive multiple-input multiple-output (MIMO) technique was employed for improving the efficiency of task offloading in MEC systems \cite{wang2020energy}.

Although exploiting the massive MIMO technique is in capable of considerably improving the offloading efficiency of MEC systems with the aid of huge beamforming gain, the associated high energy consumption and hardware cost are still obstacles in the way of its practical implementation. As a remedy, intelligent reflecting surface (IRS) has been recently proposed as a cost-effective technique to boost both the spectral-efficiency and energy-efficiency for future sixth-generation (6G) wireless networks \cite{8910627, 9326394}. Without requiring any  transmit radio-frequency (RF) chains, an IRS is a digitally-controlled meta-surface, which consists of an IRS controller and a large number of low-cost passive reflecting elements. By smartly adjusting the phase shifts of each IRS element via the IRS controller, the IRS is capable of proactively modifying the wireless propagation environment in real-time to realize different design objectiveness, e.g., signal enhancement or interference mitigation. In particular, the initial technique works \cite{8811733, wu2019beamforming} theoretically unveiled that the received power at the user is capable of scaling quadratically with the number of reflecting elements. Motivated by such a promising power scaling law, IRS has been extensively studied in the recent literature under various applications for different design objectives, such as non-orthogonal multiple-access (NOMA) \cite{Fu2021reconfigurable,9139273}, multi-cell networks \cite{9279253}, massive MIMO \cite{zhi2020power}, physical-layer security \cite{guan2020intelligent, hu2021robust}, and wireless-powered communications networks \cite{wu2021irs, 9214497, 9298890}, among others. In addition to the aforementioned applications, the IRS is also appealing for improving the task offloading efficiency of the MEC network by exploiting its high passive beamforming gain. Specifically, by properly deploying IRSs in the vicinity of IoT devices, severe distance-based signal attenuation can be effectively compensated, thus leading to a largely extended service coverage of MEC. This is of crucial importance for unlocking the full potential of the MEC in achieving high computational capability for future IoT networks.

To reap the above benefits, there are only a handful of works paying attention to investigating IRS-aided MEC systems \cite{wang2021task, chen2021irs, 9133107, 9270605, hu2021reconfigurable, zhou2020delay}. For example, a novel massive MIMO and IRS-enabled task offloading framework was proposed in \cite{wang2021task}, where the IRS beamforming, power allocation, and task offloading were jointly optimized to minimize the total energy consumption. The results in these works showed that with joint optimization of communication/computation resource allocations and IRS beamforming, the computation performance can be significantly improved. However, there are still some critical issues remain unsolved in IRS-aided MEC systems. First, all of the above works only consider the number of IRS beamforming vectors is equal to the number of devices or the IRS beamforming vector remains static in a transmission frame. Note that although employing more IRS beamforming vectors offers more degrees of freedoms (DoFs) to further enhance the system performance, it also incurs more feedback signalling overhead. This is because that the AP is in charge of the algorithmic implementation generally and then feeds the optimized IRS beamforming vectors back to the IRS controller for reconfiguring reflections due to the limited computing capability of the IRS. As such, a more flexible dynamic IRS beamforming framework is desired for striking the balance between the performance and signaling overhead. Second, all the aforementioned works focused on the partial offloading mode. Compared to partial offloading, binary offloading is easier to be practically implemented and is more appealing for some specific IoT networks where computational tasks are not partitionable. However, the efficient design of IRS-aided MEC systems adopting the binary offloading policy is currently lacking of studying.

Motivated by the above, we study in this paper an IRS-aided MEC system by considering the fundamental tradeoff between the system performance and signalling overhead incurred by dynamic IRS beamforming. Specifically, with any given number of IRS beamforming vectors available, an IRS is deployed to assist the time-division multiple access (TDMA)-based task offloading from multiple energy-constrained devices to the AP. Each device follows the binary offloading policy. To the best of the authors' knowledge, it is the first work to study the IRS-aided MEC adopting binary offloading by considering the performance-cost tradeoff. Mathematically speaking, the optimization problem under the binary offloading policy involves the hard combinatorial mode selection and thus is more challenging to be solved than that of the partial offloading. Moreover, different from the MEC system without IRS, favourable time-varying channels can be effectively generated by designing IRS beamforming vectors over different time slots (TSs). This unique DoF provided by the IRS enhances the multiuser diversity and also enables a flexible offloading design to be carried out. Note that the IRS beamforing affects the offloading decisions among devices, which makes the computational mode selection and IRS beamforming design highly coupled. The main contributions of this paper are summarized as follows:
\begin{itemize}
  \item We propose a unified dynamic IRS beamforming framework for assisting task offloading in an MEC system with binary offloading, where the IRS is allowed to adjust its passive beamforming by an arbitrary number of times. By controlling the number of IRS reconfiguration times, the system is capable of flexibly balancing the performance-cost tradeoff. This unified dynamic IRS beamforming framework generalizes the IRS beamforming configurations in previous works \cite{wang2021task, chen2021irs, 9133107, 9270605, hu2021reconfigurable, zhou2020delay} as special cases. Under the framework, our objective is to maximize the sum computation rate by jointly optimizing the computational mode selection for each device, IRS beamforming vectors, and time allocation.
  \item Due to the highly coupled optimization variables and the combinatorial nature of the multiple devices' computational mode selection, the formulated optimization problem is non-convex and NP-hard \cite{8537962}. There are generally no standard methods for solving such a problem optimally. To overcome this issue, we propose a novel penalty-based succussive convex approximation (SCA) algorithm, which includes two-layer iteration. In the inner layer, the penalized optimization problem is solved by exploiting SCA techniques, where all the variables are optimized simultaneously. While in the outer layer, the penalty factor over iterations is updated to guarantee convergence.
  \item To further reduce the computational complexity, we propose a low-complexity algorithm based on the successive refinement method. Specifically, we derive the offloading activation condition for each device by exploiting the intrinsic structure of the original problem. Utilizing the offloading activation condition, we first consider the ideal case with an asymptotically large number of IRS beamforming vectors available and develop a low computational complexity algorithm for the offloading device selection. The optimality condition of the proposed algorithm for the ideal case is revealed, which is of vital importance to characterize system performance upper bounds and provide useful guidelines for practical systems design. Inspired by the obtained solutions for the ideal case, we then reconstruct the high-quality solutions for the general case with finite number of IRS reconfiguration times. Its associated low computational complexity makes it more appealing for practical IoT networks with massive devices.
   \item Numerical results demonstrate that both the proposed two algorithms are capable of approaching the upper bound characterized by the ideal case. The maximum number of required IRS beamforming vectors for achieving the maximum computation rate performance is also revealed, which suggests the fundamental performance-cost tradeoff in exploiting the dynamic IRS beamforming. In addition, the results validate that the practical superiorities of the IRS-aided MEC over the conventional MEC without IRS in terms of supporting multiple energy-constrained devices for offloading and extending the coverage range.
\end{itemize}

The rest of this paper is organized as follows. Section II introduces the system model and problem formulation for the considered IRS-aided MEC. In Section III, we propose a penalty-based SCA algorithm. In Section IV, we propose a low-complexity algorithm based on the successive refinement method. Numerical results are presented in Section V and the paper is concluded in Section VI.

\emph{Notations:} Boldface upper-case and lower-case  letter denote matrix and   vector, respectively.  ${\mathbb C}^ {d_1\times d_2}$ stands for the set of  complex $d_1\times d_2$  matrices. For a complex-valued vector $\bf x$, ${\left\| {\bf x} \right\|}$ represents the  Euclidean norm of $\bf x$, ${\rm arg}({\bf x})$ denotes  the phase of   $\bf x$, and ${\rm diag}(\bf x) $ denotes a diagonal matrix whose main diagonal elements are extracted from vector $\bf x$.
For a vector $\bf x$, ${\bf x}^*$ and  ${\bf x}^H$  stand for  its conjugate and  conjugate transpose respectively.   For a square matrix $\bf X$,  ${\rm{Tr}}\left( {\bf{X}} \right)$, $\left\| {\bf{X}} \right\|_2$ and ${\rm{rank}}\left( {\bf{X}} \right)$ respectively  stand for  its trace, Euclidean norm and rank,  while ${\bf{X}} \succeq {\bf{0}}$ indicates that matrix $\bf X$ is positive semi-definite.
A circularly symmetric complex Gaussian random variable $x$ with mean $ \mu$ and variance  $ \sigma^2$ is denoted by ${x} \sim {\cal CN}\left( {{{\mu }},{{\sigma^2 }}} \right)$.  ${\cal O}\left(  \cdot  \right)$ is the big-O computational complexity notation.


\vspace{-8pt}
\section{System Models and Problem Formulations}
\vspace{-8pt}
\subsection{Network Model}
As shown in Fig. 1, we consider an IRS-aided MEC system, where an AP, associated with an MEC server, aims to provide edge computing services for $K$ energy-constrained devices, denoted by the set ${\cal K} \buildrel \Delta \over = \left\{ {1, \ldots ,K} \right\}$. It is assumed that the AP and all the devices are equipped with a single antenna. To overcome the high path-loss attenuation between the devices and AP, an IRS with $N$ elements, denoted by the set, ${\cal N} \buildrel \Delta \over = \left\{ {1, \ldots ,N} \right\}$, is deployed for assisting task offloading from devices to the AP. In practice, the IRS is generally implemented with low energy consumption and low cost. Thus, its computational capability is limited, which may not be able to periodically execute the algorithm. We assume that the algorithm is executed at the AP and then the optimized IRS beamforming vectors are sent from the AP to the IRS controller through a wirelessly control link. As such, the IRS is capable of dynamically adjusting the phase shift of each element over time via the IRS controller.

For the ease of practical implementation, all devices and the AP are assumed to operate over the same frequency band. Considering the quasi-static flat-fading model for all channels, the channel coefficients remain constant within the channel coherence block but may vary from one to another block. We focus on one particular frame with time duration $T$, which is chosen to be no larger than the channel coherence time and the latency of the MEC application. Within each frame, the channels from the AP to device $k$, from the AP to the IRS, and from the IRS to device $k$ are denoted by ${\bf{g}} \in \mathbb{C}^{N\times 1}$, ${{\bf{h}}_{r,k}} \in \mathbb{C}^{N \times 1}$, and ${{h}_{d,k}} \in \mathbb{C}$, respectively. The channel state information (CSI) is assumed to be perfectly known at the AP based on the channel acquisition methods discussed in \cite{9326394}, which facilitates us to characterize the performance upper bound of our considered IRS-aided MEC systems.

\begin{figure}
\begin{minipage}[t]{0.45\linewidth}
\centering
\includegraphics[width=3in]{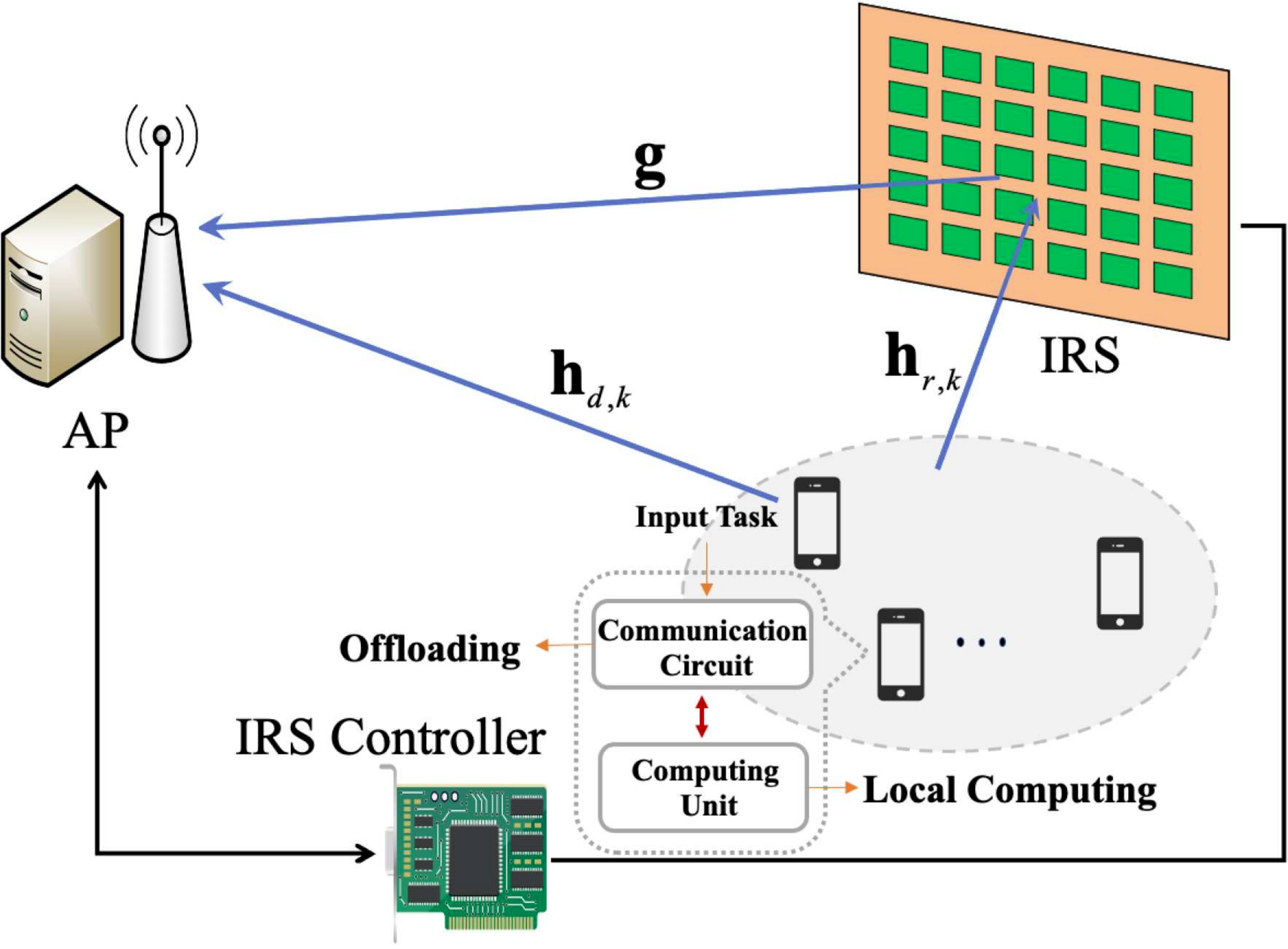}
\caption{An IRS-aided MEC system with binary offloading.}
\label{model}
\end{minipage}%
\hfill
\begin{minipage}[t]{0.45\linewidth}
\centering
\includegraphics[width=3.6in]{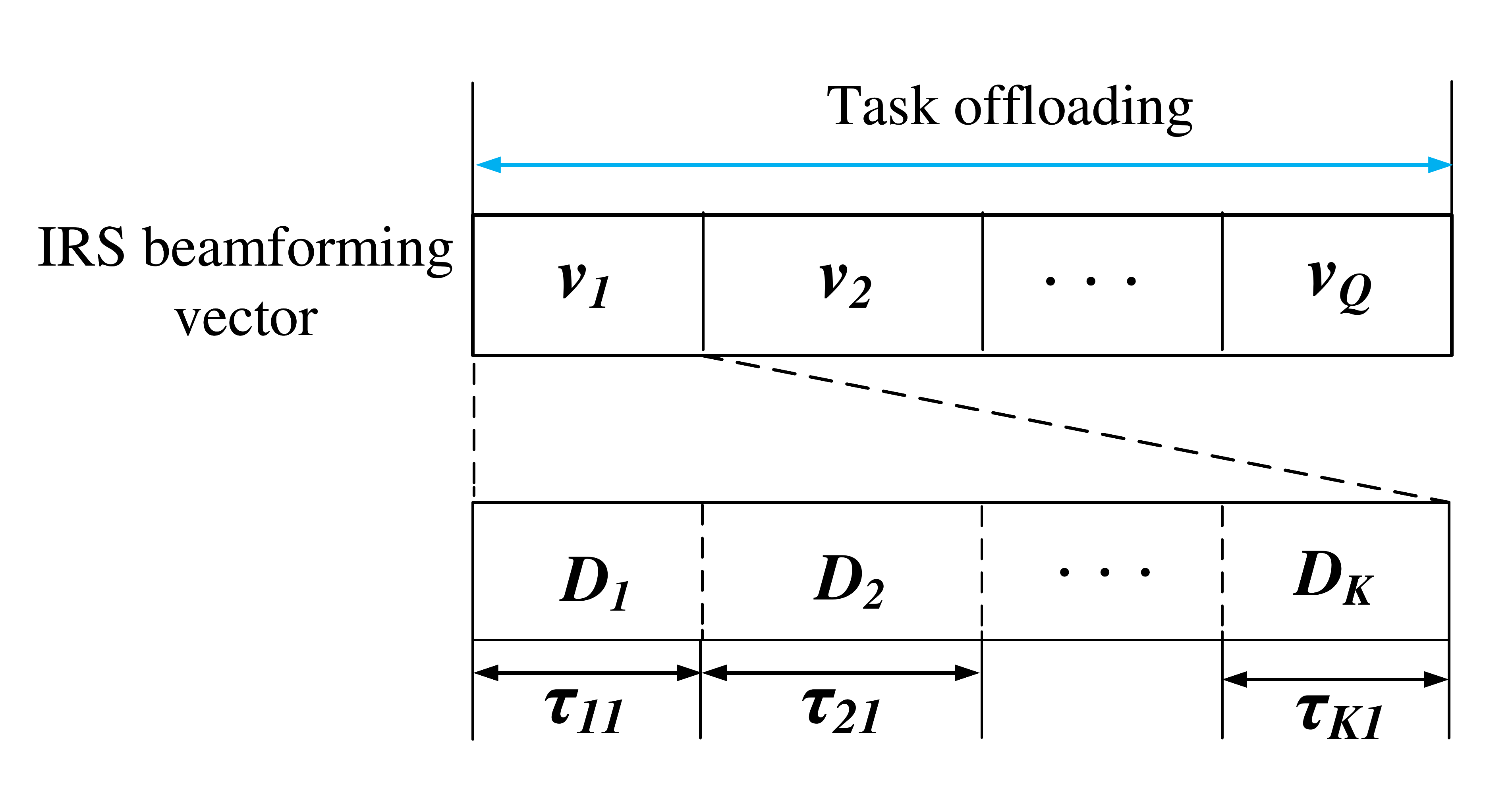}
\caption{A unified dynamic IRS beamforming framework}
\label{DIBF}
\end{minipage}\hspace{9mm}
\vspace{-16pt}
\end{figure}
\subsection{Computation and Offloading Model}
In this paper, we assume that each device adopts the binary computation-offloading policy. Specifically, the computational task must be either computed at the local device or offloaded to the MEC server, which corresponds to a wide variety of applications. For example, the measurement samples at a sensor are generally correlated in time and thereby require to be jointly processed for enhancing the estimation accuracy. We adopt two mutually exclusive sets, denoted by ${{\cal K}_{{\rm{loc}}}}$ and ${{\cal K}_{{\rm{off}}}}$, to represent the indices of the devices that operate in local computing mode and offloading mode, respectively.  It is assumed that each device has a limited energy budget dedicated for completing its own task bits, denoted as ${E_k}$ in Joules (J). Depending on the selected computation mode, the detailed operation of each device is illustrated as follows.
\subsubsection{Local Computing Mode}
When the device chooses to operate in local computing mode, it can compute throughout the entire frame of duration $T$. Let $0 \le {t_k} \le T$  and ${f_k}$ denote the computation time and the processor's operating frequency of device $k$, $k \in {{\cal K}_{{\rm{loc}}}}$, respectively, where ${f_k} \le {f_{\max }}$ holds due to the computation capability constraint. Then, the power consumption of device $k$ is ${\gamma _c}f_k^3$ (J per second), where ${\gamma _c}$ denotes the computational energy efficiency of the specific processor \cite{8304010, 8334188, 8434285}. Thus, the energy consumption of local computing is given by
\begin{align}\label{energy_loc}
E_k^{{\mathop{\rm loc}\nolimits} } = {\gamma _c}f_k^3{t_k} \le {E_k},k \in {{\cal K}_{{\rm{loc}}}},
\end{align}
Then, we use ${C_k}$ to represent the amount of required computation source for the computational tasks of device $k$, i.e., the number of CPU cycles for completing 1-bit input data. Thus, the computational bits locally executed at device $k$ in a considered frame is \cite{8304010, 8334188, 8434285}
\begin{align}\label{loc_computing}
R_k^{{\rm{loc}}} = \frac{{{t_k}{f_k}}}{{{C_k}}},k \in {{\cal K}_{{\rm{loc}}}},
\end{align}

\subsubsection{Offloading Mode}

For assisting task offloading from devices to the AP, we propose an overhead-aware dynamic IRS beamforming protocol as shown in Fig. \ref{DIBF}. Notice that the resource allocation algorithm is executed at the AP, but not at the IRS. Thus, feeding back the optimized phase shifts from the AP to the IRS before the task offloading phase, may introduce a non-negligible overhead, especially for the large number of elements. Without loss of generality, we assume that IRS beamforming vectors can be reconfigured $Q - 1$ times in total for task offloading, where $Q \ge 1$ can be any positive number, corresponding to $Q$ available IRS beamforming vectors. The $Q$ IRS beamforming matrices are denoted by ${{\bf{\Theta }}_q} = {\mathop{\rm diag}\nolimits} \left( {{e^{j{\theta _{q,1}}}}, \ldots ,{e^{j{\theta _{q,N}}}}} \right),\forall q \in {\cal Q} \buildrel \Delta \over = \left\{ {1, \ldots ,Q} \right\}$, where ${\theta _{q,n}} \in \left[ {0, 2\pi } \right), \forall n \in {\cal N}$. As such, it requires the AP to feedback $QN$ IRS phase shifts to the IRS controller, which scales linearly with $Q$. In practice, the value of $Q$ is dependent on the capacity of the control link between the AP and the IRS controller.

Regarding the task offloading, we adopt the TDMA protocol in order to avoid the interference among devices, where the orthogonal TSs are allocated to each individual device in a transmission frame. Denote ${\tau _k}$ as the time interval of the $k$-th TS, which is allocated to device $k$, $k \in {{\cal K}_{{\rm{off}}}}$. Furthermore, we assume that each device can employ any of the available $Q$ IRS beamforming vectors in order to fully unleash the potential of dynamic IRS beamforming for assisting its offloading. As such, the $k$-th TS is further partitioned into $Q$ sub-TSs. Let ${\tau _{k,q}}$ and ${p_{k,q}}$ denote the time and transmit power of device $k$ allocated to the $q$-th IRS beamforming vector, where ${\tau _k} = \sum\nolimits_{q = 1}^Q {{\tau _{k,q}}}$. Thus, the sum offloaded bits of device $k$ can be written as
\begin{align}\label{offloaded_bit}
R_k^{{\rm{off}}} &= B\sum\limits_{q = 1}^Q {{\tau _{k,q}}{{\log }_2}\left( {1 + \frac{{{p_{k,q}}{{\left| {{h_{d,k}} + {\bf{h}}_{r,k}^H{{\bf{\Theta }}_q}{\bf{g}}} \right|}^2}}}{{{\sigma ^2}}}} \right)}\nonumber\\
&= B\sum\limits_{q = 1}^Q {{\tau _{k,q}}{{\log }_2}\left( {1 + \frac{{{p_{k,q}}{{\left| {{h_{d,k}} + {\bf{q}}_k^H{{\bf{v}}_q}} \right|}^2}}}{{{\sigma ^2}}}} \right)} ,k \in {{\cal K}_{{\rm{off}}}},
\end{align}
where $B$ denotes the system bandwidth, ${{\sigma ^2}}$ is the additive Gaussian noise power at the AP, ${\bf{q}}_k^H = {\bf{h}}_{r,k}^H{\rm{diag}}\left( {\bf{g}} \right)$ and ${{\bf{v}}_q} = {\left[ {{e^{j{\theta _{q,1}}}}, \ldots ,{e^{j{\theta _{q,N}}}}} \right]^T}$.
\vspace{-10pt}
\subsection{Problem Formulation}
We aim to maximize the sum computational bits by jointly optimizing the computational mode selection of each device, the time allocation, the transmit powers, the CPU frequency, and the IRS beamforming vectors. The corresponding optimization problem can be written as
\begin{subequations}\label{C1}
\begin{align}
\label{C1-a}\mathop {\max }\limits_{\left\{ {{\tau _{k,q}}} \right\},\left\{ {{p_{k,q}}} \right\}, \left\{ {{{\bf{v}}_q}} \right\}, \left\{ {{t_k}} \right\}, \left\{ {{f_k}} \right\}, {{{\cal K}_{{\rm{off}}}}}}  \;\;&B\sum\limits_{k \in {{\cal K}_{{\rm{off}}}}} {\sum\limits_{q = 1}^Q {{\tau _{k,q}}{{\log }_2}\left( {1 + \frac{{{p_{k,q}}{{\left| {{h_{d,k}} + {\bf{q}}_k^H{{\bf{v}}_q}} \right|}^2}}}{{{\sigma ^2}}}} \right) + \sum\limits_{k \in {{\cal K}_{{\rm{loc}}}}} {\frac{{{t_k}{f_k}}}{{{C_k}}}} } } \\
\label{C1-b}{\rm{s.t.}}\;\;\;\;\;\;\;\;\;\;\;\;\;\;\;\;\;\;\;&\sum\limits_{q = 1}^Q {{\tau _{k,q}}{p_{k,q}}}  \le {E_k}, ~\forall {k} \in {{\cal K}_{{\rm{off}}}},\\
\label{C1-c}&{\gamma _c}f_k^3{t_k} \le {E_k}, ~\forall {k} \in {{\cal K}_{{\rm{loc}}}},\\
\label{C1-d}&\sum\limits_{k \in {{\cal K}_{{\rm{off}}}}} {\sum\limits_{q = 1}^Q {{\tau _{k,q}}} }  \le T,\\
\label{C1-e}&{\tau _{k,q}} \ge 0,{p_{k,q}} \ge 0, ~\forall {k} \in {{\cal K}_{{\rm{off}}}},~\forall q \in {\cal Q},\\
\label{C1-f}&0 \le {t_k} \le T, 0 \le {f_k} \le {f_{\max }}, ~\forall {k} \in {{\cal K}_{{\rm{loc}}}},\\
\label{C1-g}&\left| {{{\left[ {{{\bf{v}}_q}} \right]}_n}} \right| = 1, ~\forall n \in {\cal N},~\forall q \in {\cal Q},\\
\label{C1-h}&{{\cal K}_{{\rm{off}}}} \subseteq {\cal K},~{{\cal K}_{{\rm{loc}}}} = {\cal K}\backslash {{\cal K}_{{\rm{off}}}}.
\end{align}
\end{subequations}
In problem \eqref{C1}, constraints \eqref{C1-b} and \eqref{C1-c} ensure that the energy consumed at devices operated in either offloading or local computing mode cannot exceed their available energy. Constraints \eqref{C1-d} and \eqref{C1-e} represent the total offloading time constraint and the non-negativity constraints on the offloading power and time allocation variables, respectively. Constraint \eqref{C1-f} gives the range of tunable local computing frequencies and computing time for devices operated in local offloading mode, while constraint \eqref{C1-h} indicates the device sets of local computing and offloading modes.

Note that problem \eqref{C1} is generally a non-convex optimization problem with the combinatorial mode selection variable ${{\cal K}_{{\rm{off}}}}$ and the highly-coupled variables in both the objective function and constraints. In a conventional MEC system with binary offloading, the design of the computational mode selection is a non-trivial task since it is dependent on various practical factors, such as the channel conditions, the capability of local computing at devices, the computational complexity of tasks, and the amount of the available energy at devices. For our considered IRS aided MEC systems, the wireless channel can be proactively controlled by adjusting the reflection coefficients of the IRS, which renders the optimal mode selection is highly interacted with the optimal IRS beamforming vectors. Moreover, it is worth noting that even with a given mode selection, problem \eqref{C1} is still intractable due to the multiplicative terms in both the objective function and constraints. Therefore, there are no standard methods for solving \eqref{C1} in general.

In the following two sections, we first propose the penalty-based SCA method for solving problem \eqref{C1}, which is capable of obtaining a local-optimal solution. Then, by deeply exploiting the inherent properties of problem \eqref{C1}, a low complexity algorithm is proposed based on successive refinement method, which is more appealing for a practical scenario with massive devices.

\section{Proposed Solution Based on Penalty Method}
In this section, we first reformulate problem \eqref{C1} into a mixed integer nonlinear programming (MINLP) problem equivalently. Regarding the reformulated problem, a penalty-based algorithm is employed to solve it. Specifically, the inner layer solves a penalized optimization problem with respect to the IRS beamforming and resource allocation, while the outer layer updates the penalty factor, until the convergence is achieved.

\subsection{Problem Reformulation into a MINLP Form}
For problem \eqref{C1}, a key observation is that both the designs of the IRS beamforming and resource allocation for the devices operated in the offloading mode have no impact on the computation rate of devices operated in the local computing mode. It implies that for any given ${{{\cal K}_{loc}}}$, we can independently optimize $\left\{ {{f_k},{t_k},\forall k \in {{\cal K}_{loc}}} \right\}$ without affecting the performance of other devices. Specifically, we have the following lemma with respect to the optimal local computing frequency and time.
\begin{lem}
For ${\forall k \in {{\cal K}_{loc}}}$, the optimal local computing frequency and time, denoted by $f_k^*$ and $t_k^*$, are given by $f_k^* = \min \left( {{{\left( {\frac{{{E_k}}}{{T{\gamma _c}}}} \right)}^{\frac{1}{3}}},{f_{\max }}} \right)$, $t_k^* = T$.
\end{lem}
\begin{proof}
It is obvious that for device ${k \in {{\cal K}_{loc}}}$, it will deplete all of its available energy for maximizing its computation rate, i.e., constraint \eqref{C1-c} holds with equality, since otherwise ${{t_k}}$ or ${{f_k}}$ can be always increased to improve the objective value such that \eqref{C1-c} is active. As such, the local computation rate for device $k$ is given by
\begin{align}\label{loc_computing}
R_k^{{\rm{loc}}} = \min \left( {\frac{1}{{{C_k}}}{{\left( {\frac{{{E_k}}}{{{\gamma _c}}}} \right)}^{\frac{1}{3}}}t_k^{2/3},\frac{{{t_k}{f_{\max }}}}{{{C_k}}}} \right),\forall k \in {{\cal K}_{loc}}.
\end{align}
It can be readily checked that $R_k^{{\rm{loc}}}$ is an increasing function with respect to ${t_k}$. Hence, the maximum local computing time is achieved by setting $t_k^* = T$. Accordingly, the optimal local computing frequency can be obtained as $f_k^* = \min \left( {{{\left( {\frac{{{E_k}}}{{T{\gamma _c}}}} \right)}^{\frac{1}{3}}},{f_{\max }}} \right),\forall k \in {{\cal K}_{loc}}$.
\end{proof}

Then, we introduce a set of binary variables, denoted by $\left\{ {{b_k}} \right\}$. Specifically, ${b_k} = 0$ indicates that device $k$ performs local computing while ${b_k} = 1$ represents that device $k$ performs task offloading. By further substituting $f_k^* = \min \left( {{{\left( {\frac{{{E_k}}}{{T{\gamma _c}}}} \right)}^{\frac{1}{3}}},{f_{\max }}} \right)$ and $t_k^* = T$ into problem \eqref{C1}, problem \eqref{C1} can be reformulated in an equivalent form as follows
\begin{subequations}\label{C2}
\begin{align}
\label{C2-a}\mathop {\max }\limits_{\left\{ {{\tau _{k,q}}} \right\},\left\{ {{p_{k,q}}} \right\}, \left\{ {{{\bf{v}}_q}} \right\}, \left\{ {{b_k}} \right\}, \left\{ {r_k^{{\rm{loc}}}} \right\}}  \;\;&B\sum\limits_{k \in K} {\left( {{b_k}\sum\limits_{q = 1}^Q {{\tau _{k,q}}{{\log }_2}\left( {1 + \frac{{{p_{k,q}}{{\left| {{h_{d,k}} + {\bf{q}}_k^H{{\bf{v}}_q}} \right|}^2}}}{{{\sigma ^2}}}} \right) + r_k^{{\rm{loc}}}} } \right)} \\
\label{C2-b}{\rm{s.t.}}\;\;\;\;\;\;\;\;\;\;\;\;\;\;\;\;\;\;&{b_k}\sum\limits_{q = 1}^Q {{\tau _{k,q}}{p_{k,q}}}  \le {E_k}, ~\forall {k} \in {{\cal K}},\\
\label{C2-c}&\sum\limits_{k \in {{\cal K}}} {{b_k}\sum\limits_{q = 1}^Q {{\tau _{k,q}}} }  \le T,\\
\label{C2-d}&{b_k}{\tau _{k,q}} \ge 0,{b_k}{p_{k,q}} \ge 0, ~\forall {k} \in {{\cal K}},~\forall q \in {\cal Q},\\
\label{C2-e}&r_k^{{\rm{loc}}} \le \left( {1 - {b_k}} \right)\frac{1}{{{C_k}}}{\left( {\frac{{{E_k}}}{{{\gamma _c}}}} \right)^{\frac{1}{3}}}{T^{\frac{2}{3}}}, ~\forall {k} \in {{\cal K}},\\
\label{C2-f}&r_k^{{\rm{loc}}} \le \left( {1 - {b_k}} \right)\frac{{T{f_{\max }}}}{{{C_k}}}, ~\forall {k} \in {{\cal K}},\\
\label{C2-g}&{b_k} \in \left\{ {0,1} \right\}~\forall {k} \in {{\cal K}},\\
\label{C2-h}&\eqref{C1-g}.
\end{align}
\end{subequations}
Note that problem \eqref{C2} is a non-convex MINLP problem due to binary constraint \eqref{C2-g} and highly coupled optimization variables in objective function \eqref{C2-a} and constraints \eqref{C2-b}-\eqref{C2-d}. Nevertheless, we propose an efficient algorithm by exploiting penalty-based SCA techniques in the following for solving this challenging problem.

\subsection{Penalty-Based SCA Method}
Before solving problem \eqref{C2}, we first transform it into a more tractable form. To this end, we introduce a set of slack variables $\left\{ {{b_{k,q}}} \right\}$ and then problem \eqref{C2} can be equivalently transformed to
\begin{subequations}\label{C3}
\begin{align}
\label{C3-a}\mathop {\max }\limits_{\left\{ {{\tau _{k,q}}} \right\}, \left\{ {{{\bf{v}}_q}} \right\}, \left\{ {{b_k}} \right\}, \left\{ {{b_{k,q}}} \right\}, \left\{ {r_k^{{\rm{loc}}}} \right\}}  \;\;&B\sum\limits_{k \in {\cal K}} {\left( {\sum\limits_{q = 1}^Q {{\tau _{k,q}}{{\log }_2}\left( {1 + \frac{{{b_{k,q}}{E_k}{{\left| {{h_{d,k}} + {\bf{q}}_k^H{{\bf{v}}_q}} \right|}^2}}}{{{\tau _{k,q}}{\sigma ^2}}}} \right) + r_k^{{\rm{loc}}}} } \right)} \\
\label{C3-b}{\rm{s.t.}}\;\;\;\;\;\;\;\;\;\;\;\;\;\;\;\;\;\;&\sum\limits_{k \in {\cal K}} {\sum\limits_{q = 1}^Q {{\tau _{k,q}}} }  \le T,\\
\label{C3-c}&{\tau _{k,q}} \ge 0, {b_{k,q}} \ge 0, ~\forall {k} \in {{\cal K}},~\forall q \in {\cal Q},\\
\label{C3-d}&r_k^{{\rm{loc}}} \le \frac{1}{{{C_k}}}{\left( {\frac{{\left( {1 - {b_k}} \right){E_k}}}{{{\gamma _c}}}} \right)^{\frac{1}{3}}}{T^{\frac{2}{3}}}, ~\forall {k} \in {{\cal K}},\\
\label{C3-e}&\sum\limits_{q = 1}^Q {{b_{k,q}}}  \le {b_k}, ~\forall {k} \in {{\cal K}},\\
\label{C3-f}&\eqref{C2-f}, \eqref{C2-g}, \eqref{C1-g}.
\end{align}
\end{subequations}

Although problem \eqref{C3} is still a non-convex optimization problem, the optimization variables in constraints \eqref{C2-b}-\eqref{C2-d} are fully decoupled in constraints \eqref{C3-b} and \eqref{C3-c}. For problem \eqref{C3}, constraints \eqref{C3-b}-\eqref{C3-e} are convex and the remaining challenges for solving it are non-concave objective function \eqref{C3-a}, the set of binary variables $\left\{ {{b_k}} \right\}$ in \eqref{C2-g}, and unit-modulus constraints in \eqref{C1-g}. Regarding the binary optimization variables, we transform constraint \eqref{C2-g} into the following equality constraint as
\begin{align}\label{equality_constraint}
{b_k} - b_k^2 = 0,~k \in {\cal K}.
\end{align}
By relaxing ${b_k}$ as a continues variable, i.e., $0 \le {b_k} \le 1$, we always have ${b_k} - b_k^2 \ge 0$, where the equality holds if and only if ${b_k} = 0$ or ${b_k} = 1$, i.e., ${b_k}$ is a binary variable. In the following, we employ the principle of the penalty-based methods by integrating constraint \eqref{equality_constraint} into the objective function of problem \eqref{C3}. Specifically, we add constraint \eqref{equality_constraint} as a penalty term in objective function \eqref{C3}, yielding the following optimization problem
\begin{subequations}\label{C4}
\begin{align}
\label{C4-a}\mathop {\max }\limits_{\left\{ {{\tau _{k,q}}} \right\}, \left\{ {{{\bf{v}}_q}} \right\}, \left\{ {{b_k}} \right\}, \left\{ {{b_{k,q}}} \right\}, \left\{ {r_k^{{\rm{loc}}}} \right\}}  \;\;&B\sum\limits_{k \in {\cal K}} {\left( {\sum\limits_{q = 1}^Q {{\tau _{k,q}}{{\log }_2}\left( {1 + \frac{{{b_{k,q}}{E_k}{{\left| {{h_{d,k}} + {\bf{q}}_k^H{{\bf{v}}_q}} \right|}^2}}}{{{\tau _{k,q}}{\sigma ^2}}}} \right) + r_k^{{\rm{loc}}}} } \right)}\nonumber\\
& - \frac{1}{{2\rho }}\sum\limits_{k \in {\cal K}} {\left( {{b_k} - b_k^2} \right)}\\
\label{C4-b}{\rm{s.t.}}\;\;\;\;\;\;\;\;\;\;\;\;\;\;\;\;\;\;&0 \le {b_k} \le 1,~\forall {k} \in {{\cal K}},\\
\label{C4-c}&\eqref{C3-b}, \eqref{C3-c}, \eqref{C3-d}, \eqref{C3-e}, \eqref{C1-g}, \eqref{C2-f}, \eqref{C1-g},
\end{align}
\end{subequations}
where $\rho  > 0$ denotes the penalty factor adopted to penalize the violation of equality constraints in \eqref{equality_constraint}. It can be verified that, when $\rho  \to 0$, it follows that ${1 \mathord{\left/
 {\vphantom {1 {2\rho }}} \right.
 \kern-\nulldelimiterspace} {2\rho }} \to \infty$ and thus the solution obtained from problem \eqref{C4} satisfies equality constraint \eqref{equality_constraint}. Note that it is not effective to initialize $\rho$ to be a very small value since the objective value will be dominated by the penalty term and thus the term related to computation rate will be diminished in this case. Instead, initializing $\rho$ to be a sufficiently large value is a practically desirable way to obtain a good starting point for the proposed algorithm. By gradually decreasing the value of $\rho$, a solution that satisfies the binary constraint \eqref{equality_constraint} can be obtained within a predefined accuracy. For any given $\rho $, problem \eqref{C4} is still non-convex due to the non-concave objective function as well as unit-modules constraints in \eqref{C1-g}. Nevertheless, by applying proper change of variables, we employ SCA techniques to solve \eqref{C4} efficiently in the inner layer, where all the variables are optimized simultaneously. While in the outer layer, the penalty factor $\rho $ is updated, until the convergence is achieved.
 \subsubsection{Inner-Layer Optimization}
To deal with the non-concave objective function \eqref{C4-a}, we introduce two sets of slack variables, denoted by $\left\{ {{S_k}} \right\}$ and $\left\{ {{\eta _k}} \right\}$, and reformulate problem \eqref{C4} as follows
\begin{subequations}\label{C5}
\begin{align}
\label{C5-a}\mathop {\max }\limits_{\scriptstyle\left\{ {r_k^{{\rm{loc}}}} \right\},\left\{ {{{\bf{v}}_q}} \right\},\left\{ {{b_k}} \right\},\left\{ {{\eta _k}} \right\}\hfill\atop
\scriptstyle\left\{ {{\tau _{k,q}}} \right\},\left\{ {{S_{k,q}}} \right\},\left\{ {{b_{k,q}}} \right\}\hfill}  \;\;&B\sum\limits_{k \in {\cal K}} {\left( {\sum\limits_{q = 1}^Q {{\tau _{k,q}}{{\log }_2}\left( {1 + \frac{{{S_{k,q}}}}{{{\tau _{k,q}}{\sigma ^2}}}} \right) + r_k^{{\rm{loc}}} - \frac{1}{{2\rho }}{\eta _k}} } \right)}\\
\label{C5-b}{\rm{s.t.}}\;\;\;\;\;\;\;\;\;\;\;\;&{S_{k,q}} \le {b_{k,q}}{E_k}{\left| {{h_{d,k}} + {\bf{q}}_k^H{{\bf{v}}_q}} \right|^2},~\forall {k} \in {{\cal K}}, \forall q \in {\cal Q},\\
\label{C5-c}&{\eta _k} \ge {b_k} - b_k^2, ~\forall {k} \in {{\cal K}},\\
\label{C5-d}&\eqref{C3-b}, \eqref{C3-c}, \eqref{C3-d}, \eqref{C3-e}, \eqref{C1-g}, \eqref{C2-f}, \eqref{C4-b}.
\end{align}
\end{subequations}
Note that for the optimal solution of problem \eqref{C5}, constraints \eqref{C5-b} and \eqref{C5-c} are met with equality, since otherwise the objective value can be always increased by increasing ${S_{k,q}}$ or decreasing ${{\eta _k}}$ until \eqref{C5-b} and \eqref{C5-c} become active. As such, problem \eqref{C5} is equivalent to problem \eqref{C4}. For problem \eqref{C5}, the objective function is concave and the remaining challenges for solving it are non-convex constraints \eqref{C5-b} and \eqref{C5-c}.

Regarding constraint \eqref{C5-c}, it can be observed that its right-hand-side (RHS) is a concave function with respect to ${b_k}$, which motivates us to employ the SCA technique. For given points $\left\{ {b_k^{\left( l \right)}} \right\}$ in the $l$-th iteration, an upper bound for the RHS of \eqref{C5-c} can be obtained by using first-order Taylor expansion as follows
\begin{align}\label{upper_bound_penalty}
{b_k} - b_k^2 &\le {b_k} - {\left( {b_k^{\left( l \right)}} \right)^2} - 2b_k^{\left( l \right)}\left( {{b_k} - b_k^{\left( l \right)}} \right) \nonumber\\
& = \left( {1 - 2b_k^{\left( l \right)}} \right){b_k} + {\left( {b_k^{\left( l \right)}} \right)^2} \buildrel \Delta \over = f_k^{ub}\left( {{b_k},b_k^{\left( l \right)}} \right),\forall k \in {\cal K},
\end{align}
which is linear and convex with respect to $\left\{ {{b_k}} \right\}$.

For constraint \eqref{C5-b}, the optimization variables $\left\{ {{b_k}} \right\}$ and $\left\{ {{{\bf{v}}_q}} \right\}$ are closely coupled in its RHS. To tackle this difficulty, we apply a change of variables as follows. Let ${\left| {{h_{d,k}} + {\bf{q}}_k^H{{\bf{v}}_q}} \right|^2} = {\left| {{\bf{\bar q}}_k^H{{{\bf{\bar v}}}_q}} \right|^2}$, where ${{{\bf{\bar v}}}_q} = {\left[ {{\bf{v}}_q^H,1} \right]^H}$ and ${\bf{\bar q}}_k^H = \left[ {{\bf{q}}_k^H,{h_{d,k}}} \right]$. Define ${{{\bf{\bar Q}}}_k} = {{{\bf{\bar q}}}_k}{\bf{\bar q}}_k^H$, ${{{\bf{\bar V}}}_q} = {{{\bf{\bar v}}}_q}{\bf{\bar v}}_q^H$, which needs to satisfy ${{{\bf{\bar V}}}_q} \succeq {\bf{0}}$, ${\rm{rank}}\left( {{{{\bf{\bar V}}}_q}} \right) = 1$, and ${\left[ {{{{\bf{\bar V}}}_q}} \right]_{N + 1,N + 1}} = 1$. Thus, we can rewrite the RHS of constraint \eqref{C5-b} as
\begin{align}\label{transformation}
{b_{k,q}}{E_k}{\left| {{h_{d,k}} + {\bf{q}}_k^H{{\bf{v}}_q}} \right|^2} = {E_k}{b_{k,q}}{\mathop{\rm Tr}\nolimits} \left( {{{{\bf{\bar V}}}_q}{{{\bf{\bar Q}}}_k}} \right),~\forall {k} \in {{\cal K}}, \forall q \in {\cal Q}.
\end{align}
We note that the RHS of \eqref{transformation} is a bilinear function of the optimization variables ${b_{k,q}}$ and ${{{{\bf{\bar V}}}_q}}$, which is still non-convex. To circumvent this challenge, we further transform it into the following difference of convex (DC) functions as
\begin{align}\label{DC}
{E_k}{b_{k,q}}{\mathop{\rm Tr}\nolimits} \left( {{{{\bf{\bar V}}}_q}{{{\bf{\bar Q}}}_k}} \right) = \frac{1}{2}\left( {{{\left( {{b_{k,q}} + {\mathop{\rm Tr}\nolimits} \left( {{{{\bf{\bar V}}}_q}{{{\bf{\bar Q}}}_k}} \right)} \right)}^2} - b_{k,q}^2 - {{\left( {{\mathop{\rm Tr}\nolimits} \left( {{{{\bf{\bar V}}}_q}{{{\bf{\bar Q}}}_k}} \right)} \right)}^2}} \right), ~\forall {k} \in {{\cal K}}, \forall q \in {\cal Q}.
\end{align}
To this end, we apply the SCA method to approximate \eqref{DC} into a concave form. Specifically, for any given point $\left\{ {b_{k.q}^{\left( l \right)},{\bf{\bar V}}_q^{\left( l \right)}} \right\}$ in the $l$-th iteration, a lower bound of \eqref{DC} can be obtained as
\begin{align}\label{lower_bound_DC}
&{E_k}{b_{k,q}}{\mathop{\rm Tr}\nolimits} \left( {{{{\bf{\bar V}}}_q}{{{\bf{\bar Q}}}_k}} \right) \nonumber\\
& \ge \frac{1}{2}\left( {{{\left( {b_{k,q}^{\left( l \right)} + {\mathop{\rm Tr}\nolimits} \left( {{\bf{\bar V}}_q^{\left( l \right)}{{{\bf{\bar Q}}}_k}} \right)} \right)}^2} - b_{k,q}^2 - {{\left( {{\mathop{\rm Tr}\nolimits} \left( {{{{\bf{\bar V}}}_q}{{{\bf{\bar Q}}}_k}} \right)} \right)}^2}} \right)\nonumber\\
& + \left( {b_{k,q}^{\left( l \right)} + {\mathop{\rm Tr}\nolimits} \left( {{\bf{\bar V}}_q^{\left( l \right)}{{{\bf{\bar Q}}}_k}} \right)} \right)\left( {{b_{k,q}} - b_{k,q}^{\left( l \right)} + {\mathop{\rm Tr}\nolimits} \left( {{{{\bf{\bar V}}}_q}{{{\bf{\bar Q}}}_k}} \right) - {\mathop{\rm Tr}\nolimits} \left( {{\bf{\bar V}}_q^{\left( l \right)}{{{\bf{\bar Q}}}_k}} \right)} \right)\nonumber\\
& \buildrel \Delta \over = g_{k,q}^{lb}\left( {b_{k,q}^{\left( l \right)},{\bf{\bar V}}_q^{\left( l \right)},{b_{k,q}},{{{\bf{\bar V}}}_q}} \right), ~\forall {k} \in {{\cal K}}, \forall q \in {\cal Q},
\end{align}
where $g_{k,q}^{lb}\left( {b_{k,q}^{\left( l \right)},{\bf{\bar V}}_q^{\left( l \right)},{b_{k,q}},{{{\bf{\bar V}}}_q}} \right)$ is concave with respect to ${b_{k,q}}$ and ${{{{\bf{\bar V}}}_q}}$.

With \eqref{upper_bound_penalty} and \eqref{lower_bound_DC}, problem \eqref{C4} can be approximated as
\begin{subequations}\label{C6}
\begin{align}
\label{C6-a}\mathop {\max }\limits_{\scriptstyle\left\{ {r_k^{{\rm{loc}}}} \right\},\left\{ {{{{{\bf{\bar V}}}_q}}} \right\},\left\{ {{b_k}} \right\},\left\{ {{\eta _k}} \right\}\hfill\atop
\scriptstyle\left\{ {{\tau _{k,q}}} \right\},\left\{ {{S_{k,q}}} \right\},\left\{ {{b_{k,q}}} \right\}\hfill}  \;\;&B\sum\limits_{k \in {\cal K}} {\left( {\sum\limits_{q = 1}^Q {{\tau _{k,q}}{{\log }_2}\left( {1 + \frac{{{S_{k,q}}}}{{{\tau _{k,q}}{\sigma ^2}}}} \right) + r_k^{{\rm{loc}}} - \frac{1}{{2\rho }}{\eta _k}} } \right)}\\
\label{C6-b}{\rm{s.t.}}\;\;\;\;\;\;\;\;\;\;\;\;&{S_{k,q}} \le g_{k,q}^{lb}\left( {b_{k,q}^{\left( l \right)},{\bf{\bar V}}_q^{\left( l \right)},{b_{k,q}},{{{\bf{\bar V}}}_q}} \right),~\forall {k} \in {{\cal K}}, \forall q \in {\cal Q},\\
\label{C6-c}&{\eta _k} \ge f_k^{ub}\left( {{b_k},b_k^{\left( l \right)}} \right), ~\forall {k} \in {{\cal K}},\\
\label{C6-d}&{\left[ {{{{\bf{\bar V}}}_q}} \right]_{n,n}} = 1, ~n = 1, \ldots ,N + 1, \\
\label{C6-e}&{\rm{rank}}\left( {{{{\bf{\bar V}}}_q}} \right) = 1,\\
\label{C6-f}&\eqref{C3-b}, \eqref{C3-c}, \eqref{C3-d}, \eqref{C3-e}, \eqref{C2-f}, \eqref{C4-b}.
\end{align}
\end{subequations}
Problem \eqref{C6} becomes a convex semidefinite program (SDP) by relaxing the rank-one constraint \eqref{C6-e}. As such, we can successively solve it with the help of standard convex optimization techniques, such as CVX, until convergence is achieved. However, the obtained ${{{{\bf{\bar V}}}_q}}$ may not satisfy the rank-one constraint. Although Gaussian randomization can be employed to construct a rank-one solution after convergence, it may not be in capable of guaranteeing globally and/or locally optimal solution. To overcome this issue, we transform the rank-one constraint \eqref{C6-e} into the following equivalence
\begin{align}\label{rank_one}
{\rm{rank}}\left( {{{{\bf{\bar V}}}_q}} \right) = 1 \Leftrightarrow {\mathop{\rm Tr}\nolimits} \left( {{{{\bf{\bar V}}}_q}} \right) - {\left\| {{{{\bf{\bar V}}}_q}} \right\|_2} = 0.
\end{align}
Then, the resulting problem can be efficiently solved by using the DC-based penalty framework \cite{Fu2021reconfigurable}. The details are omitted due to its brevity.
\subsubsection{Outer Layer Iteration}
In the outer layer, we gradually decrease the value of the penalty factor $\rho$ as follows
\begin{align}\label{update_penalty}
\rho : = c\rho, ~0 < c < 1,
\end{align}
where $c$ denotes the step size. The selection of the value for $c$ strikes a balance between the performance and the associated number of iterations. Generally, a larger value of $c$ can achieve better performance but at the cost of more iterations in the outer layer.
\subsubsection{Overall Algorithm and Computational Complexity Analysis}
Recall that the equality constraint \eqref{equality_constraint} for problem \eqref{C3} needs to be satisfied in the converged solution. To this end, for evaluating whether the obtained solution violates the equality constraint or not, we introduce a indicator $\xi$ defined as
\begin{align}\label{indicator}
\xi  = \max \left\{ {{b_k} - b_k^2,\forall k \in {\cal K}} \right\}.
\end{align}
The algorithm is terminated when $\xi  \le \chi$, where $\chi$ is a predefined threshold for characterizing the accuracy of the equality constraint. The details of the proposed algorithm are summarized in Algorithm 1. Note that in the inner layer of our proposed algorithm, a local optimal solution can be obtained since all the optimization variables are optimized simultaneously in an iterative way by using SCA. Therefore, based on the results in \cite{Shi2016joint}, the obtained solution converges to a point satisfying the Karush-Kuhn-Tucker (KKT) conditions of the original problem.

The complexity of Algorithm 1 can be summarized as follows. In the inner layer, the main complexity comes from step 5. Specifically, \eqref{C6} can be solved by the interior-point method, whose complexity is given by ${\cal O}\left( {{{\left( {3K\left( {Q + 1} \right) + QN} \right)}^{3.5}}} \right)$. As such, the total complexity of Algorithm 1 is ${\cal O}\left( {{L_{outer}}{L_{inner}}{{\left( {3K\left( {Q + 1} \right) + QN} \right)}^{3.5}}} \right)$, where ${{L_{outer}}}$ and ${{L_{inner}}}$ represent the number of iterations required for reaching convergence in the outer layer and inner layer, respectively.
\begin{algorithm}[!t]
 \caption{Penalty-based algorithm for solving problem \eqref{C1}.}
 \label{alg1}
 \begin{algorithmic}[1]
  \STATE  \textbf{Initialize} ${{\bf{\bar V}}_q^{\left( 0 \right)}}$, $b_k^{\left( 0 \right)}$, $b_{k,q}^{\left( 0 \right)}$, $c$, $\rho$, $\varepsilon$, and  $\chi$.
  \STATE  \textbf{repeat: outer layer}
  \STATE \quad Set iteration index $l = 0$ for inner layer.
  \STATE \quad \textbf{repeat: inner layer }
  \STATE  \qquad For given $\left\{ {{\bf{\bar V}}_q^{\left( l \right)},b_k^{\left( l \right)},b_{k,q}^{\left( l \right)}} \right\}$, solve problem \eqref{C6}.
  \STATE  \qquad Update $\left\{ {{\bf{\bar V}}_q^{\left( l+1 \right)},b_k^{\left( l+1 \right)},b_{k,q}^{\left( l+1 \right)}} \right\}$ with the obtained solution, and $l = l + 1$.
  \STATE \quad \textbf{until} the fractional increase of the objective value of \eqref{C6} is below a threshold $\varepsilon$.
  \STATE  \quad Update penalty factor $\rho$ based on \eqref{update_penalty}.
  \STATE \textbf{until}  penalty  violation $\xi$ is below a threshold $\chi$.
 \end{algorithmic}
\end{algorithm}
\begin{rem}
Although in this paper, we focus on IRS-aided MEC systems with binary offloading, the proposed Algorithm 1 can be directly extended to a similar system with partial offloading. Suppose an MEC system employing dynamic IRS beamforming with partial offloading, the corresponding optimization problem can be formulated from problem \eqref{C3} by replacing constraint \eqref{C2-g} with a relaxed version, i.e., $0 \le {b_k} \le 1,\forall k \in {\cal K}$, which is convex. The remaining issues are the the non-concave objective function as well as unit-modules constraints, which can be tackled in a similar manner as previously discussed (i.e., \eqref{C5-a}, \eqref{transformation}, \eqref{DC}, and \eqref{lower_bound_DC}). As such, the computation rate maximization problem for an MEC system employing dynamic IRS beamforming with partial offloading can still be solved by the proposed algorithm.
\end{rem}

\section{Low Complexity Solution by Succussive Refinement}
Although Algorithm 1 proposed in the previous section yields high-quality solution for the original problem \eqref{C1}, its computational complexity can still be high with increasing $K$, $Q$, and $N$. To overcome this issue, in this section, we propose a more computationally efficient algorithm for problem \eqref{C1} by deeply exploiting its inherent structures. In the following subsections, we first derive a sufficient and necessary condition for determining whether a device should activate offloading. Then, based on the offloading activation condition, we propose a succussive refinement-based algorithm for solving a special case of problem \eqref{C1}, where the number of IRS reconfiguration times is asymptotically large. Finally, the proposed algorithm for the special case is extended to solve the original problem \eqref{C1}.
\vspace{-12pt}
\subsection{When to Activate Offloading?}
Before answering the question regarding when to activate offloading for each device, we first provide the following lemma to shed light on how the $\left| {{{\cal K}_{{\rm{off}}}}} \right|$ users make use of the $Q$ available IRS beamforming vectors for assisting task offloading.
\begin{lem}
Denote the optimal IRS beamforming vectors for problem \eqref{C1} by ${\bf{v}}_q^*,q \in {\cal Q}$. Then, problem \eqref{C1} is equivalent to the following problem
\begin{subequations}\label{C7}
\begin{align}
\label{C7-a}\mathop {\max }\limits_{\left\{ {{\tau _k}} \right\}, {{{\cal K}_{{\rm{off}}}}}}  \;\;&B\sum\limits_{k \in {{\cal K}_{{\rm{off}}}}} {{\tau _k}{{\log }_2}} \left( {1 \!\!+\!\! \frac{{{E_k}{{\left| {{h_{d,k}} + {\bf{q}}_k^H{\bf{v}}_{\Omega \left( k \right)}^*} \right|}^2}}}{{{\tau _k}{\sigma ^2}}}} \right) \!\!+\!\! \sum\limits_{k \in {{\cal K}_{{\rm{loc}}}}} {\min \left( {\frac{{T{f_{\max }}}}{{{C_k}}},\frac{{{T^{2/3}}{{\left( {\frac{{{E_k}}}{{{\gamma _c}}}} \right)}^{1/3}}}}{{{C_k}}}} \right)} \\
\label{C7-b}{\rm{s.t.}}\;\;\;\;&\sum\nolimits_{k \in {{\cal K}_{{\rm{off}}}}} {{\tau _k} \le T},\\
\label{C7-e}&{\tau _k} \ge 0, ~\forall {k} \in {{\cal K}_{{\rm{off}}}},\\
\label{C7-h}&{{\cal K}_{{\rm{off}}}} \subseteq {\cal K},~{{\cal K}_{{\rm{loc}}}} = {\cal K}\backslash {{\cal K}_{{\rm{off}}}},
\end{align}
\end{subequations}
where $\Omega \left( k \right) = \arg \mathop {\max }\limits_{q \in {\cal Q}} {\left| {{h_{d,k}} + {\bf{q}}_k^H{\bf{v}}_q^*} \right|^2}$.
\end{lem}
\begin{proof}
The main procedures are similar to our previous work \cite{wu2021irs}. Interested reader may refer to \cite{wu2021irs} for more details.
\end{proof}

Lemma 2 provides an important insight into the optimal association between the given IRS beamforming vectors and the offloading devices. In particular, each device only needs to occupy one IRS beamforming vector for task offloading and uses up all of its energy. Based on Lemma 2, we further derive a sufficient and necessary condition for determining whether an arbitrary device should be scheduled for task offloading. The results are shown in the following theorem.
\begin{thm}\label{thm1}
(Offloading Activation Condition): Suppose that the optimal IRS beamforming vectors for problem \eqref{C4} are given by (${\bf{v}}_q^*$'s), $q \in {\cal Q}$. Given any activate offloading device set ${{\cal K}_{{\rm{off}}}} \subseteq {\cal K}$ and any not activate offloading device set ${{\cal K}_{{\rm{loc}}}} = {\cal K}\backslash {{\cal K}_{{\rm{off}}}}$, the device $k' \in {{\cal K}_{{\rm{loc}}}}$ can be activated for offloading and added to ${{\cal K}_{{\rm{off}}}}$ if and only if
\begin{align}\label{offloading_condition}
&BT{\log _2}\left( {\frac{{T{\sigma ^2} + {E_{k'}}{{\left| {{h_{d,k'}} + {\bf{q}}_{k'}^H{\bf{v}}_{\Omega \left( {k'} \right)}^*} \right|}^2} + \sum\nolimits_{k \in {{\cal K}_{{\rm{off}}}}} {{E_k}{{\left| {{h_{d,k}} + {\bf{q}}_k^H{\bf{v}}_{\Omega \left( k \right)}^*} \right|}^2}} }}{{T{\sigma ^2} + \sum\nolimits_{k \in {{\cal K}_{{\rm{off}}}}} {{E_k}{{\left| {{h_{d,k}} + {\bf{q}}_k^H{\bf{v}}_{\Omega \left( k \right)}^*} \right|}^2}} }}} \right)\nonumber\\
&\ge\min \left( {\frac{{T{f_{\max }}}}{{{C_{k'}}}},\frac{{{T^{2/3}}}}{{{C_{k'}}}}{{\left( {\frac{{{E_{k'}}}}{{{\gamma _c}}}} \right)}^{1/3}}} \right).
\end{align}
\end{thm}
\begin{proof}
Please refer to Appendix A.
\end{proof}

The interpretation of Theorem 1 is that if the new device $k'$ wants to join the offloading device set, the improvement of computation rate attained by offloading should be higher than that brought by local computing. Lemma 2 and Theorem 1 serve as foundation for developing the low-complexity algorithm for the original problem \eqref{C1}.

However, the optimal IRS beamforming vectors ${\bf{v}}_q^*$'s for assisting offloading remain unknown yet and the explicit association relationship between IRS beamforming vectors and offloading devices is not clear. In the next subsection, we first consider a special case of problem \eqref{C1}, where the number of IRS reconfiguration times is asymptotically large, i.e., $Q \to \infty$. Then, the proposed solution is further extended to the general case of any finite value $Q$.

\vspace{-12pt}
\subsection{Special Case Study: $Q \to \infty$}
\vspace{-4pt}
In this subsection, we consider a special case of problem \eqref{C4} with asymptotically large $Q$, which serves as the performance upper bound of the general case with finite $Q$. The results in Lemma 2 imply that at most $\left| {{{\cal K}_{{\rm{off}}}}} \right|$ IRS beamforming vectors are sufficient for achieving the maximum offloading rate when the number of available IRS beamforming vectors is unlimited. As such, problem \eqref{C1} is equivalently reduced to the following problem
\begin{subequations}\label{C20}
\begin{align}
\label{C20-a}\mathop {\max }\limits_{\left\{ {{\tau _k}} \right\}, {{{\cal K}_{{\rm{off}}}}}, \left\{ {{{\bf{v}}_k}} \right\}}  \;&B\sum\limits_{k \in {{\cal K}_{{\rm{off}}}}} {{\tau _k}{{\log }_2}} \left( {1 \!+\! \frac{{{E_k}{{\left| {{h_{d,k}} \!\!+\!\! {\bf{q}}_k^H{{\bf{v}}_k}} \right|}^2}}}{{{\tau _k}{\sigma ^2}}}} \right) \!+\! \sum\limits_{k \in {{\cal K}_{{\rm{loc}}}}} {\min \left( {\frac{{T{f_{\max }}}}{{{C_k}}},\frac{{{T^{2/3}}{{\left( {\frac{{{E_k}}}{{{\gamma _c}}}} \right)}^{1/3}}}}{{{C_k}}}} \right)} \\
\label{C20-b}{\rm{s.t.}}\;\;\;\;\;\;\;&\sum\nolimits_{k \in {{\cal K}_{{\rm{off}}}}} {{\tau _k} \le T},\\
\label{C20-c}&{\tau _k} \ge 0, ~\forall {k} \in {{\cal K}_{{\rm{off}}}},\\
\label{C20-d}&{{\cal K}_{{\rm{off}}}} \subseteq {\cal K},~{{\cal K}_{{\rm{loc}}}} = {\cal K}\backslash {{\cal K}_{{\rm{off}}}},\\
\label{C20-e}&\left| {{{\left[ {{{\bf{v}}_k}} \right]}_n}} \right| = 1, ~n \in {\cal N}.
\end{align}
\end{subequations}

For problem \eqref{C20}, it is observed that ${{\bf{v}}_k}$'s are only appealed in the objective function and each of them is exclusively separated in its own achievable offloading rate. Therefore, the optimal solution for $\left\{ {{{\bf{v}}_k}} \right\}$ can be obtained by solving $K$ optimization problems independently in parallel. Specifically, for device $k$, $\forall k \in {{\cal K}_{{\rm{off}}}}$, the corresponding problem with respect to ${{{\bf{v}}_k}}$ can be expressed as
\begin{subequations}\label{C9}
\begin{align}
\label{C9-a}\mathop {\max }\limits_{{{{\bf{v}}_k}}}  \;\;&{{{\left| {{h_{d,k}} + {\bf{q}}_k^H{{\bf{v}}_k}} \right|}^2}} \\
\label{C9-b}{\rm{s.t.}}\;\;&\left| {{{\left[ {{{\bf{v}}_k}} \right]}_n}} \right| = 1, ~n \in {\cal N}.
\end{align}
\end{subequations}
The optimal solution for \eqref{C9} can be obtained in closed form expression, which is given by
\begin{align}\label{opt_beamforming}
{\bf{v}}_k^* = \exp \left( { - j\left( {\arg \left( {{\bf{q}}_k^H} \right) - \arg \left( {{h_{d,k}}} \right)} \right)} \right), ~\forall k \in {{\cal K}_{{\rm{off}}}}.
\end{align}
Define ${g_k} = {\left| {{h_{d,k}} + {\bf{q}}_k^H{\bf{v}}_k^*} \right|^2}$ and then problem \eqref{C20} can be simplified as
\begin{subequations}\label{C21}
\begin{align}
\label{C21-a}\mathop {\max }\limits_{\left\{ {{\tau _k}} \right\}, {{{\cal K}_{{\rm{off}}}}}}  \;&B\sum\limits_{k \in {{\cal K}_{{\rm{off}}}}} {{\tau _k}{{\log }_2}} \left( {1 + \frac{{{E_k}{g_k}}}{{{\tau _k}{\sigma ^2}}}} \right) + \sum\limits_{k \in {{\cal K}_{{\rm{loc}}}}} {\min \left( {\frac{{T{f_{\max }}}}{{{C_k}}},\frac{{{T^{2/3}}{{\left( {\frac{{{E_k}}}{{{\gamma _c}}}} \right)}^{1/3}}}}{{{C_k}}}} \right)} \\
\label{C21-b}{\rm{s.t.}}\;\;\;&\sum\nolimits_{k \in {{\cal K}_{{\rm{off}}}}} {{\tau _k} \le T},\\
\label{C21-c}&{\tau _k} \ge 0, ~\forall {k} \in {{\cal K}_{{\rm{off}}}},\\
\label{C21-d}&{{\cal K}_{{\rm{off}}}} \subseteq {\cal K},~{{\cal K}_{{\rm{loc}}}} = {\cal K}\backslash {{\cal K}_{{\rm{off}}}}.
\end{align}
\end{subequations}

It can be verified that objective function \eqref{C21-a} is concave and constraints \eqref{C21-b}-\eqref{C21-c} are convex. The remaining challenge for solving problem \eqref{C21} is how to determine the offloading device set ${{\cal K}_{{\rm{off}}}}$. Fortunately, Theorem 1 explicitly gives the condition on whether a new device should be activated for offloading, which motivates us to propose an efficient successive refinement algorithm to determine ${{\cal K}_{{\rm{off}}}}$ iteratively. Specifically, we first introduce the concept of the trading computation rate as follows
\begin{align}\label{trading_rate}
{\Lambda _k} = BT{\log _2}\left( {1 + \frac{{{E_k}{g_k}}}{{T{\sigma ^2}}}} \right) - \min \left( {\frac{{T{f_{\max }}}}{{{C_k}}},\frac{{{T^{2/3}}{{\left( {{E_k}/{\gamma _c}} \right)}^{1/3}}}}{{{C_k}}}} \right), ~\forall k \in {\cal K},
\end{align}
which represents the difference between the computation rate attained by offloading and local computing for each device. Then, we sort all devices in descending order, i.e., ${\Lambda _{\Phi \left( 1 \right)}} \ge {\Lambda _{\Phi \left( 2 \right)}} \ge  \ldots  \ge {\Lambda _{\Phi \left( K \right)}}$, and set the initial ${{\cal K}_{{\rm{off}}}} = \emptyset$. Then, we successively take one device from the order and decide whether it can be added to ${{\cal K}_{{\rm{off}}}}$ according to Theorem 1. Finally, the device set ${{\cal K}_{{\rm{off}}}^ \star}$ is obtained when one device is found unable to be added to ${{\cal K}_{{\rm{off}}}}$ or all devices are added to ${{\cal K}_{{\rm{off}}}}$. The details of this procedure are summarized in Algorithm 2.
\begin{algorithm}[!t]
 \caption{Successive refinement algorithm for solving problem \eqref{C20}.}
 \label{alg1}
 \begin{algorithmic}[1]
  \STATE  Compute the trading computation rate ${\Lambda _k}$ for $\forall k \in {\cal K}$ by \eqref{trading_rate}.
  \STATE  Sort the trading computation rate of all devices in the descending order, i.e., ${\Lambda _{\Phi \left( 1 \right)}} \ge {\Lambda _{\Phi \left( 2 \right)}} \ge  \ldots  \ge {\Lambda _{\Phi \left( K \right)}}$, and set ${{\cal K}_{{\rm{off}}}} = \emptyset$.
  \STATE \textbf{for} $k = 1:K$
  \STATE  \qquad \textbf{if} inequality \eqref{offloading_condition} is satisfied ~\textbf{do}
  \STATE  \qquad \qquad ${{\cal K}_{{\rm{off}}}} = {{\cal K}_{{\rm{off}}}} \cup \left\{ {\Phi \left( k \right)} \right\}$.
  \STATE \qquad  \textbf{else}
  \STATE \qquad \qquad Stop and output ${{\cal K}_{{\rm{off}}}}$.
  \STATE \qquad\textbf{end if}
  \STATE \textbf{end for}
  \STATE For the obtained ${{\cal K}_{{\rm{off}}}}$, solve problem \eqref{C21}.
 \end{algorithmic}
\end{algorithm}

To understand Algorithm 2 better, we provide the following proposition to characterize its optimal condition.
\begin{proposition}
Algorithm 2 is optimal for problem \eqref{C20} when ${E_1} = {E_2} =  \ldots  = {E_K}$ and ${C_1} = {C_2} =  \ldots  = {C_K}$.
\end{proposition}
\begin{proof}
Please refer to Appendix B.
\end{proof}

Proposition 1 reveals that Algorithm 2 is in capable of achieving the maximum computation rate in a homogeneous MEC scenario, where the available energy and the intensity of computational tasks for all devices are identical, i.e., ${E_1}  =  \ldots  = {E_K}$ and ${C_1}  =  \ldots  = {C_K}$. This can be interpreted as follows. In a homogeneous MEC scenario, the computation rates brought by local computing are the same for all devices. In this case, it can be shown that a device with higher trading computation rate is equivalent to that the device has larger channel power gain for the offloading link. As such, the system tends to activate the $\left| {{{\cal K}_{{\rm{off}}}^ \star}} \right|$ devices for offloading who enjoy the highest channel power gains based on Theorem 1.

\subsection{General Case with Finite $Q$}
In this subsection, we consider the general case with finite value $Q$. Based on the solution derived for the special case with $Q \to \infty $, we provide an efficient approach to reconstruct the solution for the case with finite $Q$.

It is worth noting that if the number of offloading devices obtained by Algorithm 2 is smaller than $Q$, i.e., $\left| {{{\cal K}_{{\rm{off}}}^ \star}} \right| \le Q$, its solution is also applicable to the general case with finite $Q$. However, for the case that $\left| {{{\cal K}_{{\rm{off}}}^ \star}} \right| > Q$, the number of available IRS beamforming vectors may not sufficient for enabling each offloading device to occupy its dedicated IRS beamforming vector. Accordingly, multiple devices may share the same IRS beamforming vector for task offloading. As such, all devices in task offloading mode need to be partitioned into $Q$ disjoint groups, in which the devices belong to the same group employ the common IRS beamforming vector for assisting offloading. Unfortunately, the user grouping method problem is generally intractable since the optimal grouping strategy requires an exhaustive search for all the possible cases, i.e., ${Q^{\left| {{{\cal K}_{{\rm{off}}}}} \right|}}$, which results in an exponential computational complexity. Thus, it may be prohibitive in practice, especially in an overloaded scenario, i.e., $K$ is practically large.

To overcome the aforementioned challenge, we propose an efficient scheme by exploiting the solution for the special case with infinite $Q$. Specifically, we suppose that each device can be occupied on its dedicated IRS beamforming vector as described in the previous subsection and the corresponding channel power gain is ${g_k} = {\left| {{h_{d,k}} + {\bf{q}}_k^H{\bf{v}}_k^*} \right|^2}$. First, all the devices are sorted in descending order according to their trading computation rate defined in \eqref{trading_rate}. We use $\Phi \left( k \right)$ to denote the index for the device of the $k$-th order. Second, each of the first $Q-1$ devices, denoted by $\Phi \left( 1 \right), \ldots \Phi \left( {Q - 1} \right)$, is assigned with its dedicated IRS beamforming vector and the rest $\left| {{{\cal K}_{{\rm{off}}}}} \right| - Q+1$ devices share a common IRS beamforming vector for offloading. Let $\Psi  = \left\{ {\Omega \left( 1 \right), \ldots ,\Omega \left( {Q - 1} \right)} \right\}$. Based on this scheme, problem \eqref{C7} is transformed into following optimization problem for any given ${{{\cal K}_{{\rm{off}}}}}$:
\begin{subequations}\label{C30}
\begin{align}
\label{C30-a}\mathop {\max }\limits_{\left\{ {{\tau _k}} \right\}, {{{\bf{v}}_Q}}}  \;\;&\sum\limits_{q = 1}^{Q - 1} {{\tau _{\Omega \left( q \right)}}{{\log }_2}\left( {1 + \frac{{{E_{\Omega \left( q \right)}}{g_{\Omega \left( q \right)}}}}{{{\tau _{\Omega \left( q \right)}}{\sigma ^2}}}} \right)}  + \sum\limits_{k \in {{\cal K}_{{\rm{off}}}}\backslash \Psi } {{\tau _k}{{\log }_2}\left( {1 + \frac{{{E_k}{{\left| {{h_{d,k}} + {\bf{q}}_k^H{{\bf{v}}_Q}} \right|}^2}}}{{{\tau _k}{\sigma ^2}}}} \right)}\\
\label{C30-b}{\rm{s.t.}}\;\;\;\;&\sum\nolimits_{k \in {{\cal K}_{{\rm{off}}}}} {{\tau _k} \le T},\\
\label{C30-e}&{\tau _k} \ge 0, ~\forall {k} \in {{\cal K}_{{\rm{off}}}},\\
\label{C30-h}&\left| {{{\left[ {{{\bf{v}}_Q}} \right]}_n}} \right| = 1,n \in {\cal N}.
\end{align}
\end{subequations}
According to the results in Theorem 1, problem \eqref{C30} can be equivalently transformed into a more tractable form as follows
\begin{subequations}\label{C31}
\begin{align}
\label{C31-a}\mathop {\max }\limits_{{{{\bf{v}}_Q}}}  \;\;&T{\log _2}\left( {1 + \sum\nolimits_{q = 1}^{Q - 1} {\frac{{{E_{\Omega \left( q \right)}}{g_{\Omega \left( q \right)}}}}{{T{\sigma ^2}}} + \sum\nolimits_{k \in {{\cal K}_{{\rm{off}}}}\backslash \Psi } {\frac{{{E_k}{{\left| {{h_{d,k}} + {\bf{q}}_k^H{{\bf{v}}_Q}} \right|}^2}}}{{T{\sigma ^2}}}} } } \right)\\
\label{C31-b}{\rm{s.t.}}\;\;\;&\eqref{C30-h}.
\end{align}
\end{subequations}
\begin{algorithm}[!t]
 \caption{Successive refinement algorithm for solving problem \eqref{C1}.}
 \label{alg1}
 \begin{algorithmic}[1]
  \STATE  Compute the trading computation rate ${\Lambda _k}$ for $\forall k \in {\cal K}$ by \eqref{trading_rate}.
  \STATE  Sort the trading computation rate of all devices in the descending order, i.e., ${\Lambda _{\Phi \left( 1 \right)}} \ge {\Lambda _{\Phi \left( 2 \right)}} \ge  \ldots  \ge {\Lambda _{\Phi \left( K \right)}}$, and set ${{\cal K}_{{\rm{off}}}} = \emptyset$.
  \STATE \textbf{for} $k = 1:K$
  \STATE  \qquad \textbf{if} $k \le Q-1$ ~\textbf{do}
  \STATE  \qquad \qquad  Update ${{\bf{v}}_k}$ by maximizing ${\left| {{h_{d,\Omega \left( k \right)}} + {\bf{q}}_{\Omega \left( k \right)}^H{{\bf{v}}_{\Omega \left( k \right)}}} \right|^2}$
  \STATE  \qquad \textbf{else}
    \STATE  \qquad \qquad  Update ${{\bf{v}}_Q}$ by solving problem \eqref{C32}
  \STATE \qquad\textbf{end if}
  \STATE  \qquad  \textbf{if} inequality \eqref{offloading_condition} is satisfied ~\textbf{do}
  \STATE  \qquad \qquad ${{\cal K}_{{\rm{off}}}} = {{\cal K}_{{\rm{off}}}} \cup \left\{ {\Phi \left( k \right)} \right\}$.
  \STATE \qquad  \textbf{else}
  \STATE \qquad \qquad Stop and output ${{\cal K}_{{\rm{off}}}}$.
  \STATE \qquad\textbf{end if}
  \STATE \textbf{end for}
  \STATE For the obtained ${{\cal K}_{{\rm{off}}}}$ and $\left\{ {{{\bf{v}}_q}} \right\}$, solve problem \eqref{C21}.
 \end{algorithmic}
\end{algorithm}
One observation for problem \eqref{C31} is that the optimization variable ${{{\bf{v}}_Q}}$ only appeals in the term $\sum\nolimits_{k \in {{\cal K}_{{\rm{off}}}}\backslash \Psi } {{E_k}{{\left| {{h_{d,k}} + {\bf{q}}_k^H{{\bf{v}}_Q}} \right|}^2}/\left( {T{\sigma ^2}} \right)}$. As such, we can focus on the following simplified optimization problem instead:
\begin{subequations}\label{C32}
\begin{align}
\label{C32-a}\mathop {\max }\limits_{{{{\bf{v}}_Q}}}  \;\;&{\sum\nolimits_{k \in {{\cal K}_{{\rm{off}}}}\backslash \Psi } {\frac{{{E_k}{{\left| {{h_{d,k}} + {\bf{q}}_k^H{{\bf{v}}_Q}} \right|}^2}}}{{T{\sigma ^2}}}} }\\
\label{C32-b}{\rm{s.t.}}\;\;\;&\eqref{C30-h}.
\end{align}
\end{subequations}
Although problem \eqref{C32} is non-convex due to the convex objective function and unit-modulus constraints, the convexity of \eqref{C31-a} motivates us to employ the SCA technique for solving it. Specifically, taking the first-order Taylor expansion of ${{{\left| {{h_{d,k}} + {\bf{q}}_k^H{{\bf{v}}_Q}} \right|}^2}}$ at any feasible point ${{{\bf{\hat v}}}_Q}$, we obtain
\begin{align}\label{gain_lower}
{\left| {{h_{d,k}} + {\bf{q}}_k^H{{\bf{v}}_Q}} \right|^2} &\ge  - {\left| {{h_{d,k}} + {\bf{q}}_k^H{{{\bf{\hat v}}}_Q}} \right|^2} + 2{\mathop{\rm Re}\nolimits} \left( {{{\left( {{h_{d,k}} + {\bf{q}}_k^H{{{\bf{\hat v}}}_Q}} \right)}^H}\left( {{h_{d,k}} + {\bf{q}}_k^H{{\bf{v}}_Q}} \right)} \right)\nonumber\\
& = {D_k} + 2{\mathop{\rm Re}\nolimits} \left( {\left( {{\bf{\hat v}}_Q^H{{\bf{q}}_k}{\bf{q}}_k^H + h_{d,k}^H{\bf{q}}_k^H} \right){{\bf{v}}_Q}} \right),
\end{align}
where ${D_k} =  - {\left| {{h_{d,k}} + {\bf{q}}_k^H{{{\bf{\hat v}}}_Q}} \right|^2} + 2{\mathop{\rm Re}\nolimits} \left( {{{\left( {{h_{d,k}} + {\bf{q}}_k^H{{{\bf{\hat v}}}_Q}} \right)}^H}{h_{d,k}}} \right)$. Then, a lower bound of objective function \eqref{C31-a} is given by
\begin{align}\label{sum_gain_lower}
\sum\limits_{k \in {{\cal K}_{{\rm{off}}}}\backslash \Psi } {\frac{{{E_k}{D_k}}}{{T{\sigma ^2}}}}  + 2{\mathop{\rm Re}\nolimits} \left( {\left( {\sum\limits_{k \in {{\cal K}_{{\rm{off}}}}\backslash \Psi } {\frac{{{E_k}\left( {{\bf{\hat v}}_Q^H{{\bf{q}}_k}{\bf{q}}_k^H + h_{d,k}^H{\bf{q}}_k^H} \right)}}{{T{\sigma ^2}}}} } \right){{\bf{v}}_Q}} \right).
\end{align}
The optimal solution satisfying \eqref{C30-h} for maximizing \eqref{sum_gain_lower} can be derived as
\begin{align}\label{v_Q}
{{\bf{v}}_Q} = \exp \left( {j\arg \left( {\sum\limits_{k \in {{\cal K}_{{\rm{off}}}}\backslash \Psi } {\frac{{{E_k}\left( {{{\bf{q}}_k}{\bf{q}}_k^H{{{\bf{\hat v}}}_Q} + {{\bf{q}}_k}{h_{d,k}}} \right)}}{{T{\sigma ^2}}}} } \right)} \right).
\end{align}
Then, problem \eqref{C32} can be efficiently solved by maximizing its lower bound \eqref{sum_gain_lower} in an iterative manner until the convergence is achieved.

Finally, we can extend the successive refinement method for solving original problem \eqref{C1} with finite $Q$. The main steps are similar to those for solving the problem with infinite $Q$. One difference can be summarized as follows. For the first $Q-1$ devices, i.e., $\Psi  = \left\{ {\Omega \left( 1 \right), \ldots ,\Omega \left( {Q - 1} \right)} \right\}$, each device would be assigned with its dedicated IRS beamforming vector. While for the remaining devices belonging to ${\cal K}\backslash \Psi$, they would share the common IRS beamforming vector, i.e., ${{\bf{v}}_Q}$ when adding to ${{\cal K}_{{\rm{off}}}}$. The details of the procedure are summarized in Algorithm 3. The complexity of Algorithm 3 mainly lies in ordering the devices and calculating the shared IRS beamforming vector for problem \eqref{C32}. The overall computational complexity is given by ${\cal O}\left( {K{{\log }_2}\left( K \right) + N{L_{iter}}} \right)$, where $L_{iter}$ denotes the number of iterations required for solving problem \eqref{C32}. Compared to Algorithm 1, its associated computational complexity is significantly reduced, thus making it more appealing for practical systems with large $K$ and $N$.
\vspace{-10pt}
\section{Numerical Results}
In this section, we provide numerical results to validate the effectiveness of the proposed designs and to draw useful insights into IRS-aided MEC systems with binary offloading. We consider a three dimensional coordinate setup, where the AP and the IRS are located at $\left( {0,0,0} \right)$ meter (m) and $\left( {30,0,4} \right)$ m, respectively, and all the devices are uniformly and randomly distributed within a radius of 4 m centered at $\left( {30,0,0} \right)$ m. The distance-dependent path-loss model is given by $L\left( d \right) = {c_0}{\left( {d/{d_0}} \right)^{ - \alpha }}$, where ${c_0} =  - 30$ dB represents the signal attenuation at a reference distance of ${d_0} = 1$ m, $d$ is the link distance, and $\alpha$ is the path-loss exponent. The path-loss exponents for the AP-IRS link, the IRS-device link, and the AP-device link are set to 2.2, 2.2, and 3.4, respectively. We assume that all the links follow Rician fading with a Rician factor of 3 dB. Unless otherwise stated, other system parameters are set as follows: $N = 60$, $T=1$ s, $B = 1$ MHz, ${\sigma ^2} =  - 80$ dBm, ${\gamma _c} = {10^{ - 28}}$, ${E_k} = 10$ dBm, $\forall k$, $Q = 5$, and $C = {C_k} = \left[ {500,1000,1500} \right]$ cycles/s, $\forall k$.
\vspace{-10pt}
\subsection{Impact of Dynamic IRS beamforming}
In Fig. \ref{Dynamic_beamforming}, we study the impact of dynamic IRS beamforming on MEC systems employing binary offloading, by plotting the average sum computation rate and the number of offloading devices versus the number of IRS beamforming vectors available, i.e., $Q$, respectively. For comparison, we consider the following schemes: 1) \textbf{Upper bound}: the succussive refinement method presented in Algorithm 2 for solving the corresponding problem with asymptotically large $Q$; 2) \textbf{Proposed Algorithm 1}: the penalty-based SCA method presented in Algorithm 1 for solving problem \eqref{C1}; 3) \textbf{Proposed Algorithm 3}: the succussive refinement method presented in Algorithm 3 for solving problem \eqref{C1}.
\begin{figure*}[t!]
\centering
\subfigure[Average sum computation rate versus $Q$.]{\label{sumrate_versus_Q}
\includegraphics[width= 2.5in, height=2in]{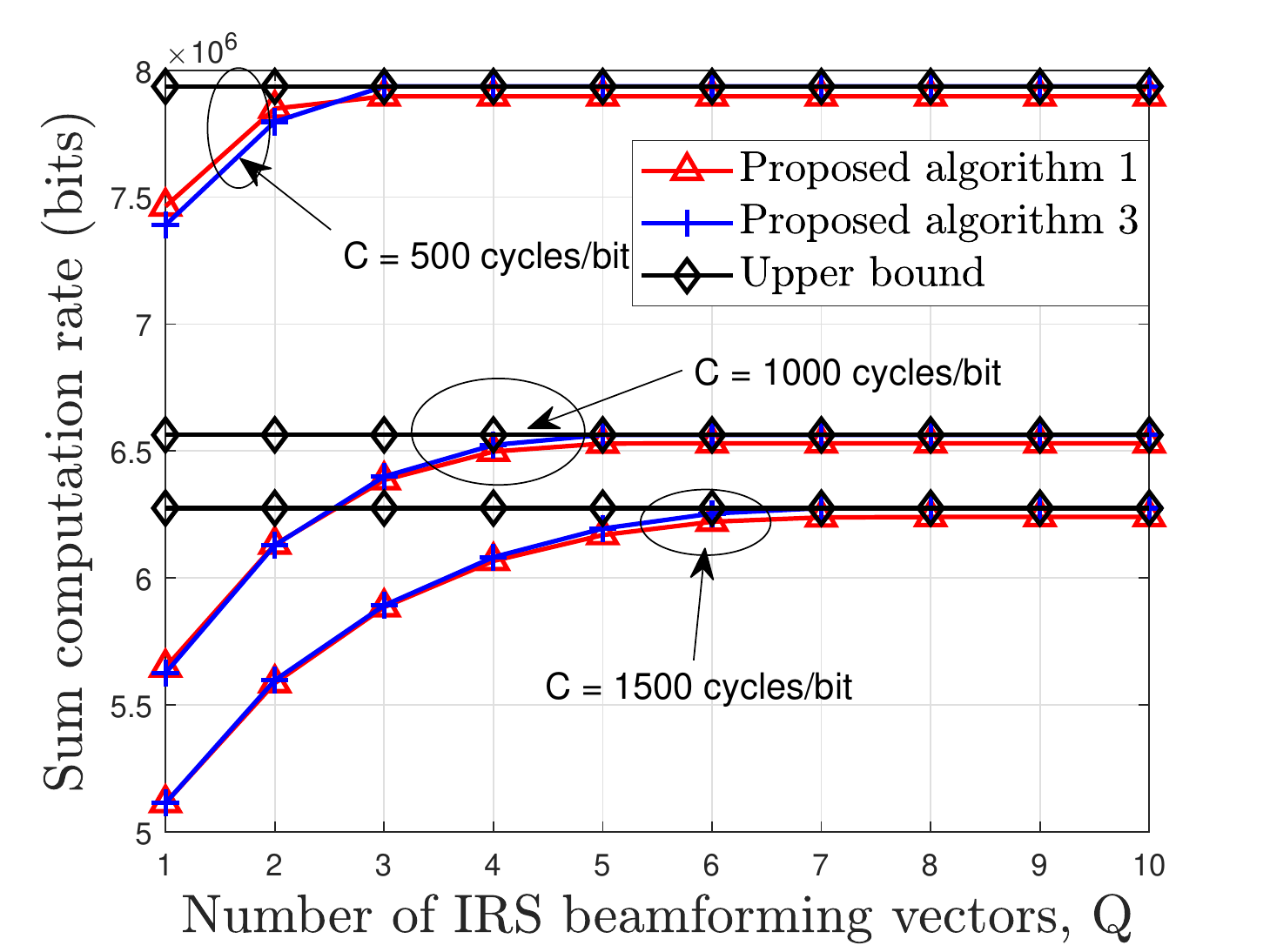}}
\subfigure[Average $\left| {{{\cal K}_{{\rm{off}}}}} \right|$ versus $Q$]{\label{devices_versus_Q}
\includegraphics[width= 2.5in, height=2in]{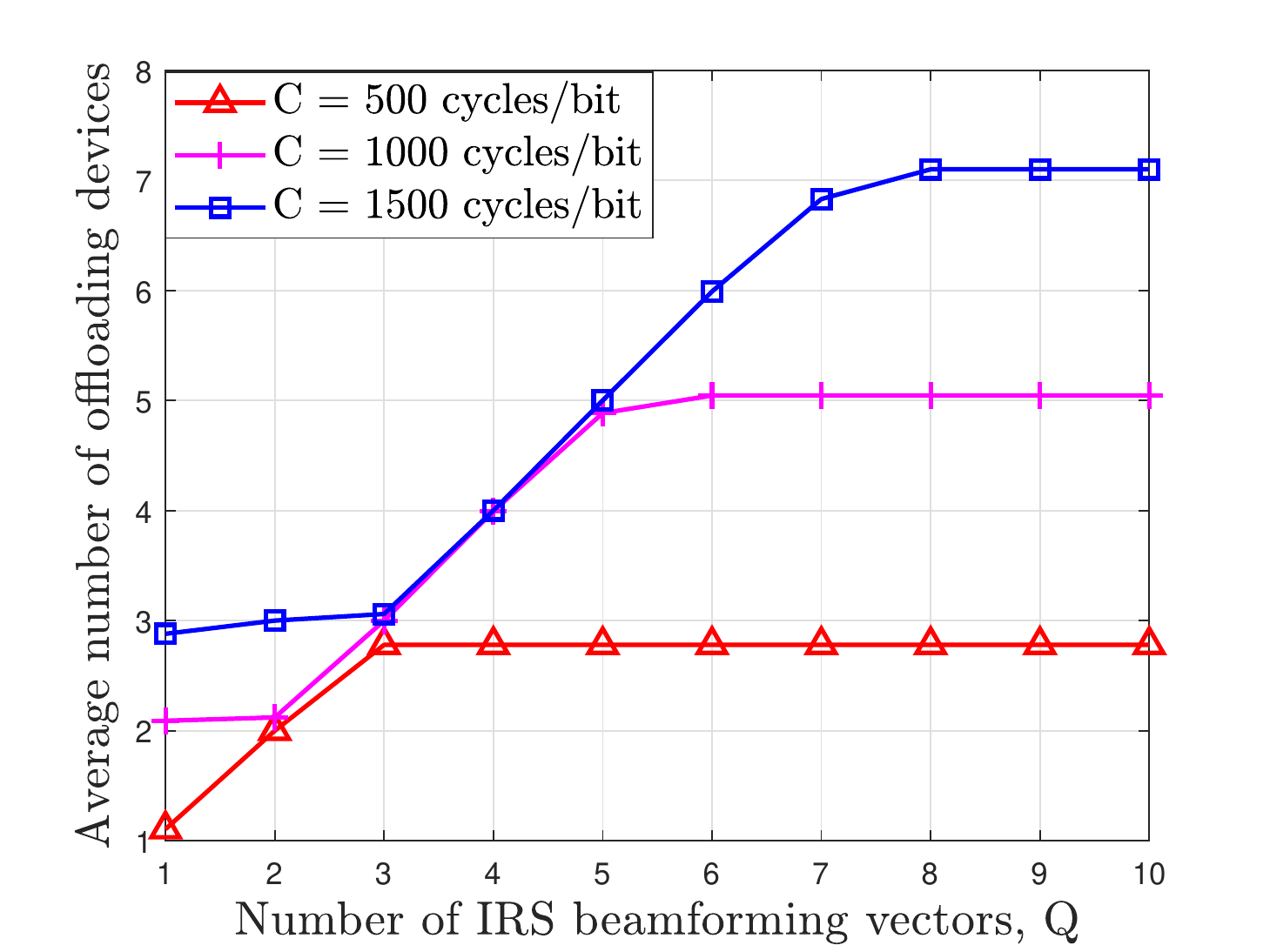}}
\setlength{\abovecaptionskip}{0.4cm}
\caption{{Impact of dynamic IRS beamforming on MEC systems.}}\label{Dynamic_beamforming}
\vspace{-16pt}
\end{figure*}

It is first observed from Fig. \ref{sumrate_versus_Q} that for all cases of $C = 500$ cycles/bit, $C = 1000$ cycles/bit, and $C = 1500$ cycles/bit, the average sum computation rate first increases and then remains constant as $Q$ increases, which indicates that employing dynamic IRS beamforming is indeed beneficial for the computation rate improvement of IRS-aided MEC systems. Second, it can be shown that our proposed Algorithm 3 can achieve the similar performance as the Algorithm 1 and can approach the optimal performance characterized by the upper bound when $Q$ is larger. This validates the effectiveness of our proposed algorithms and also verifies the generality of proposed dynamic IRS beamforming optimization framework.

Although exploiting more IRS beamforming vectors provides the higher design flexibility for achieving the higher computation rate, one can observe that when $Q$ is larger than a certain value, the computation rate remains constant as $Q$ increases. In particular, the required number of IRS beamforming vectors for cases of $C = 500$ cycles/bit, $C = 1000$ cycles/bit, and $C = 1500$ cycles/bit to achieve the maximum computation rate is given by 3, 5, and 7, respectively. It is worth pointing out that the number of IRS phase-shift coefficients required to be feed back from the AP to the IRS controller is $QN$, which increases linearly with $Q$. The result implies a fundamental performance-overhead tradeoff in employing the dynamic IRS beamforming. Thus, the value of $Q$ should be carefully determined based on the system parameters.

Finally, it is observed from Fig. \ref{devices_versus_Q} that as $Q$ increases, the number of offloading devices first gradually increases and then remains constant, which means that more devices can be activated for task offloading with the help of dynamic IRS beamforming. This further demonstrates the effectiveness of dynamic IRS beamforming for unlocking the potential computational capability at edge server, besides its ability for improving the sum computation rate. Moreover, it can be observed that the maximum number of offloading devices for cases of $C = 500$ cycles/bit, $C = 1000$ cycles/bit, and $C = 1500$ cycles/bit is given by 3, 5, and 7, respectively, which is consistent with the required number of IRS beamforming vectors for the corresponding three cases presented in Fig. \ref{sumrate_versus_Q}. The result demonstrates that the number of offloading devices is highly impacted by $C$. As shown in Theorem 1, the device tends to offload when $C$ becomes higher. The reason is that the computation rate attained by local computing is negligible at large value of $C$, which forces more devices to perform task offloading for exploiting computational resources at edge servers.

Note that Algorithm 3 is capable of achieving similar performance as Algorithm 1 and significantly reducing computational complexity, which renders that Algorithm 3 is more appealing for practical systems with large $K$ and $N$. Therefore, the performance evaluation in the next subsection is based on Algorithm 3 if not otherwise specified.
\vspace{-10pt}
\subsection{Performance Comparison}
In order to show the sum computation rate gain brought by the IRS in the MEC system, we compare the following schemes: 1) \textbf{Binary offloading}: problem \eqref{C1} is solved by adopting the successive refinement method presented in Algorithm 3; 2) \textbf{Offloading only}: the IRS beamforming and communications resource allocation are jointly optimized when allowing all the devices to offload their tasks to the AP, i.e., ${{\cal K}_{{\rm{off}}}} = {\cal K}$; 3) \textbf{Random IRS beamforming}: Each phase shift of the IRS is randomly and uniformly distributed over $\left( {0,2\pi } \right]$
while other optimization variables in problem \eqref{C1} are optimized; 4) \textbf{Without IRS}: Binary offloading is considered without adopting IRS; 5) \textbf{Offloading without IRS}: All the devices offload their tasks to the AP without adopting IRS; 6) \textbf{Local computing}: All the devices perform computations locally, i.e., ${{\cal K}_{{\rm{loc}}}} = {\cal K}$.

\begin{figure*}[t!]
\centering
\subfigure[C = 500 cycles/bit.]{\label{sumrate_500}
\includegraphics[width= 2.5in, height=2in]{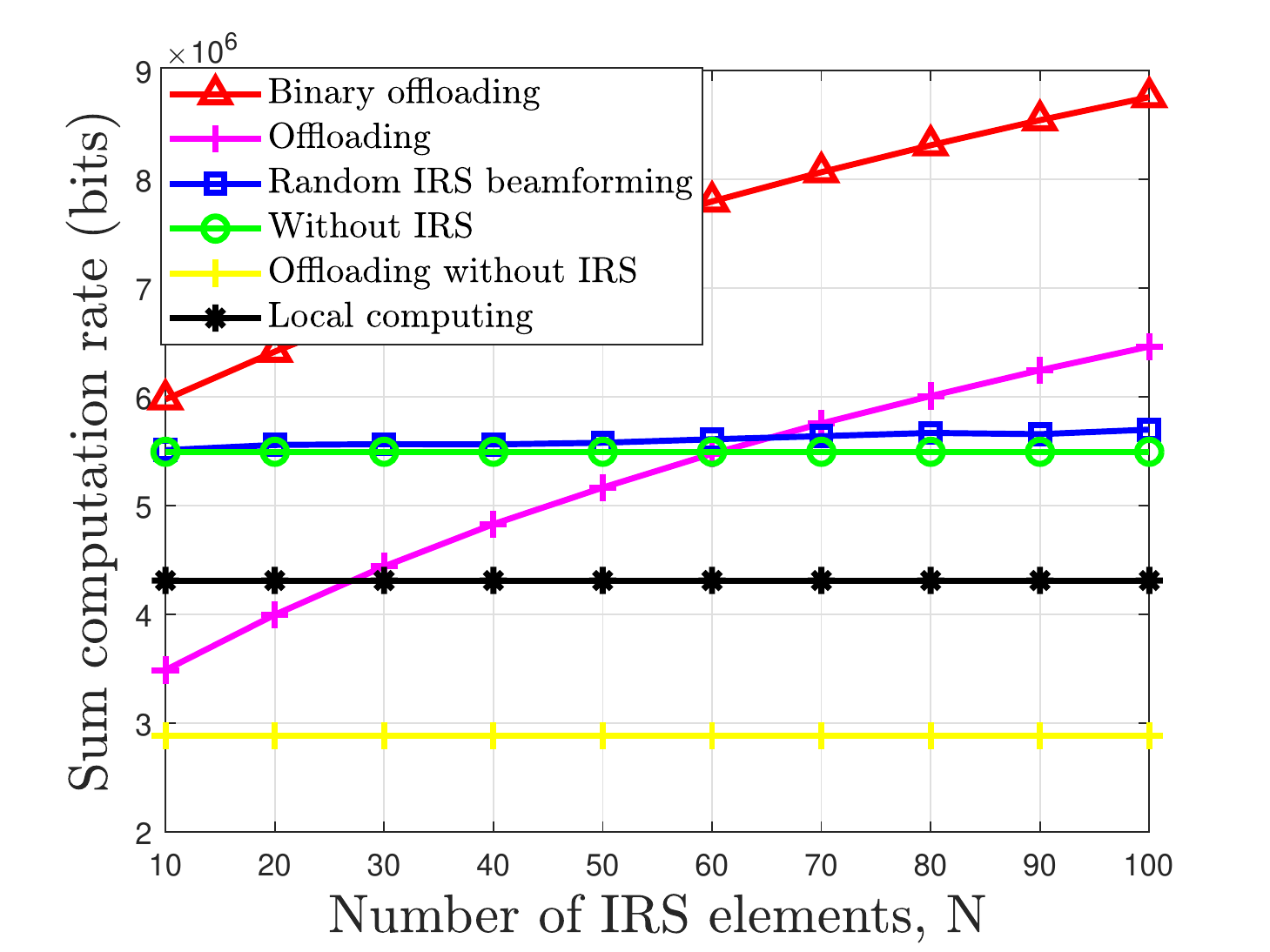}}
\subfigure[C = 1000 cycles/bit.]{\label{sumrate_1000}
\includegraphics[width= 2.5in, height=2in]{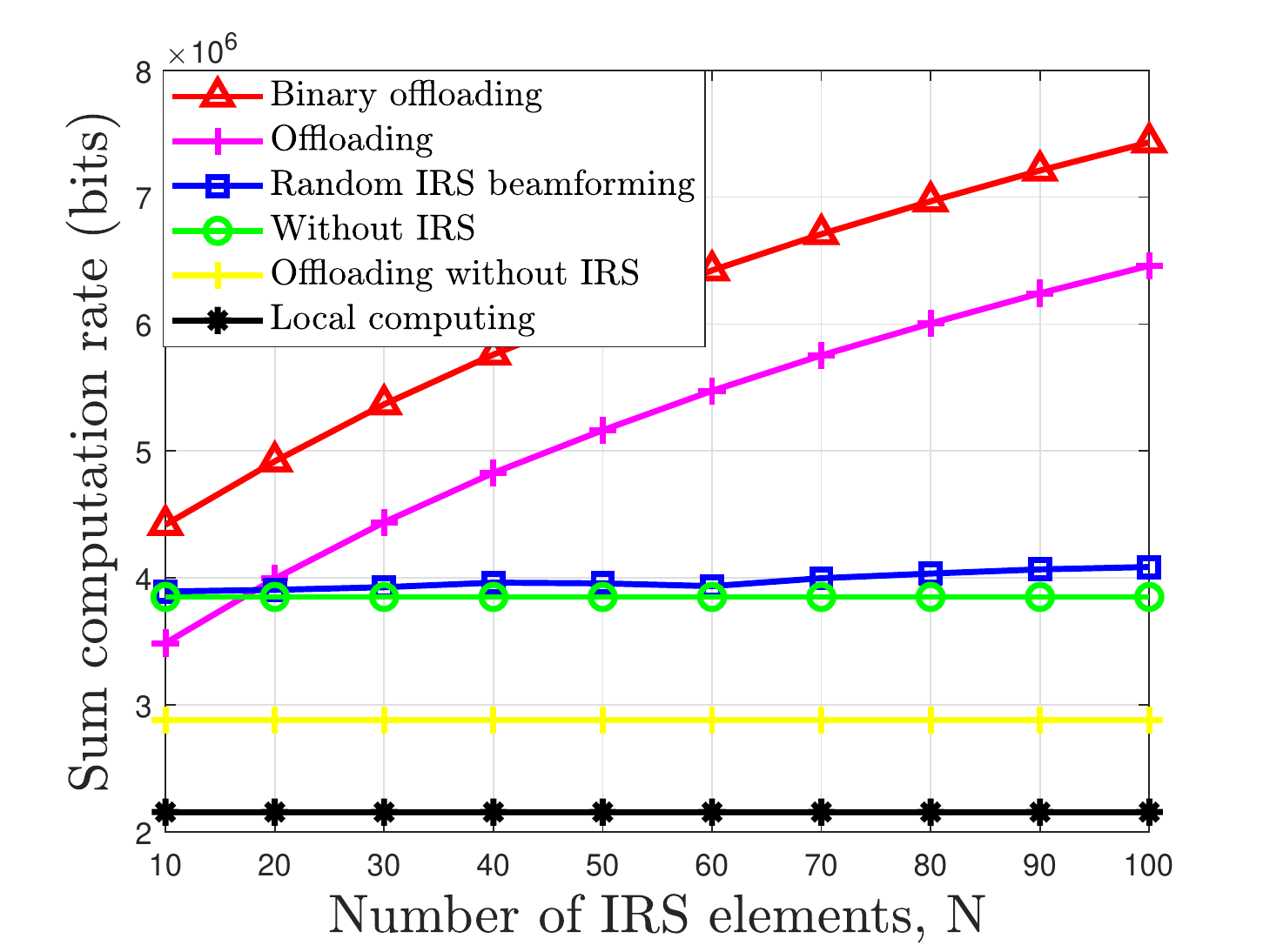}}
\setlength{\abovecaptionskip}{0.4cm}
\caption{{Average sum computation rate versus $N$.}}\label{Rate_versus_N}
\vspace{-16pt}
\end{figure*}

\subsubsection{Impact of Number of IRS Elements}

In Fig. \ref{Rate_versus_N}, we compare the sum computation rate obtained by all the schemes versus $N$ under two cases of $C = 500$ cycles/bit and $C = 1000$ cycles/bit, respectively. For both the two cases, it is observed from Fig. \ref{Rate_versus_N} that the sum computation rate gain achieved by our proposed designs over the benchmark schemes increases as $N$ since more IRS elements provide higher passive beamforming gain for assisting offloading. In addition, the performance of the scheme with random IRS beamforming only attains a marginal gain over the system without IRS, whereas the scheme with offloading only performs even worse than the system with random IRS beamforming, but outperforms it for large $N$. This is expected since the IRS beamforming gain is more dominant as $N$ increases, which is beneficial for compensating the performance loss incurred by computing modes selection. Moreover, one can observe that the scheme of offloading without IRS performs worse than local computing at small $C$. By increasing the number of IRS elements, the performance of offloading only is capable of significantly outperforming local computing, which further demonstrates the usefulness of deploying IRS for unleashing the potential of MEC servers.
\subsubsection{Impact of Distance Between AP and Device Center}
\begin{figure}
\begin{minipage}[t]{0.45\linewidth}
\centering
\includegraphics[width=2.9in]{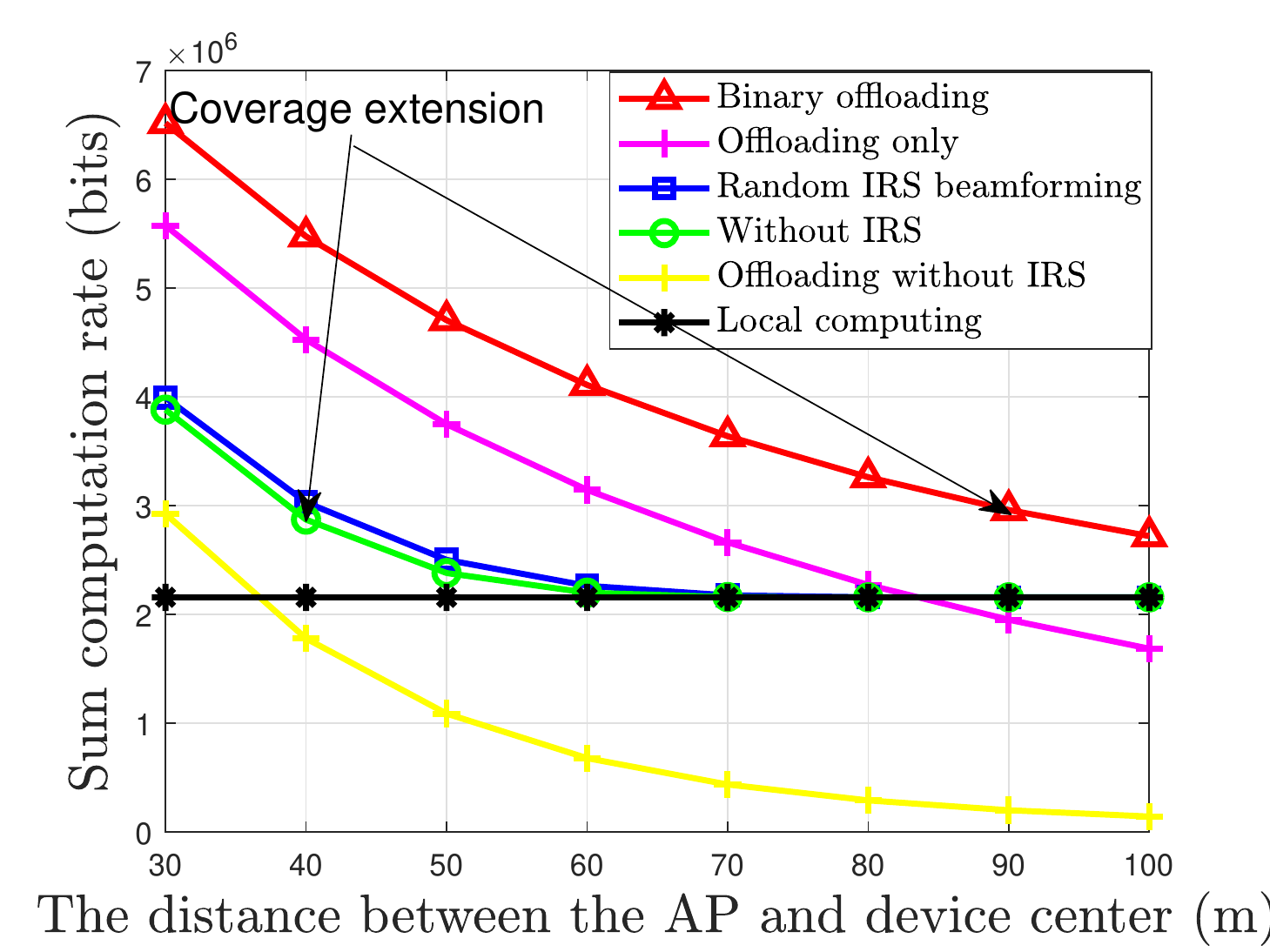}
\caption{Sum computation rate versus the distance between AP and device center with C = 1000 cycles/bit.}
\label{distance}
\end{minipage}%
\hfill
\begin{minipage}[t]{0.45\linewidth}
\centering
\includegraphics[width=2.9in]{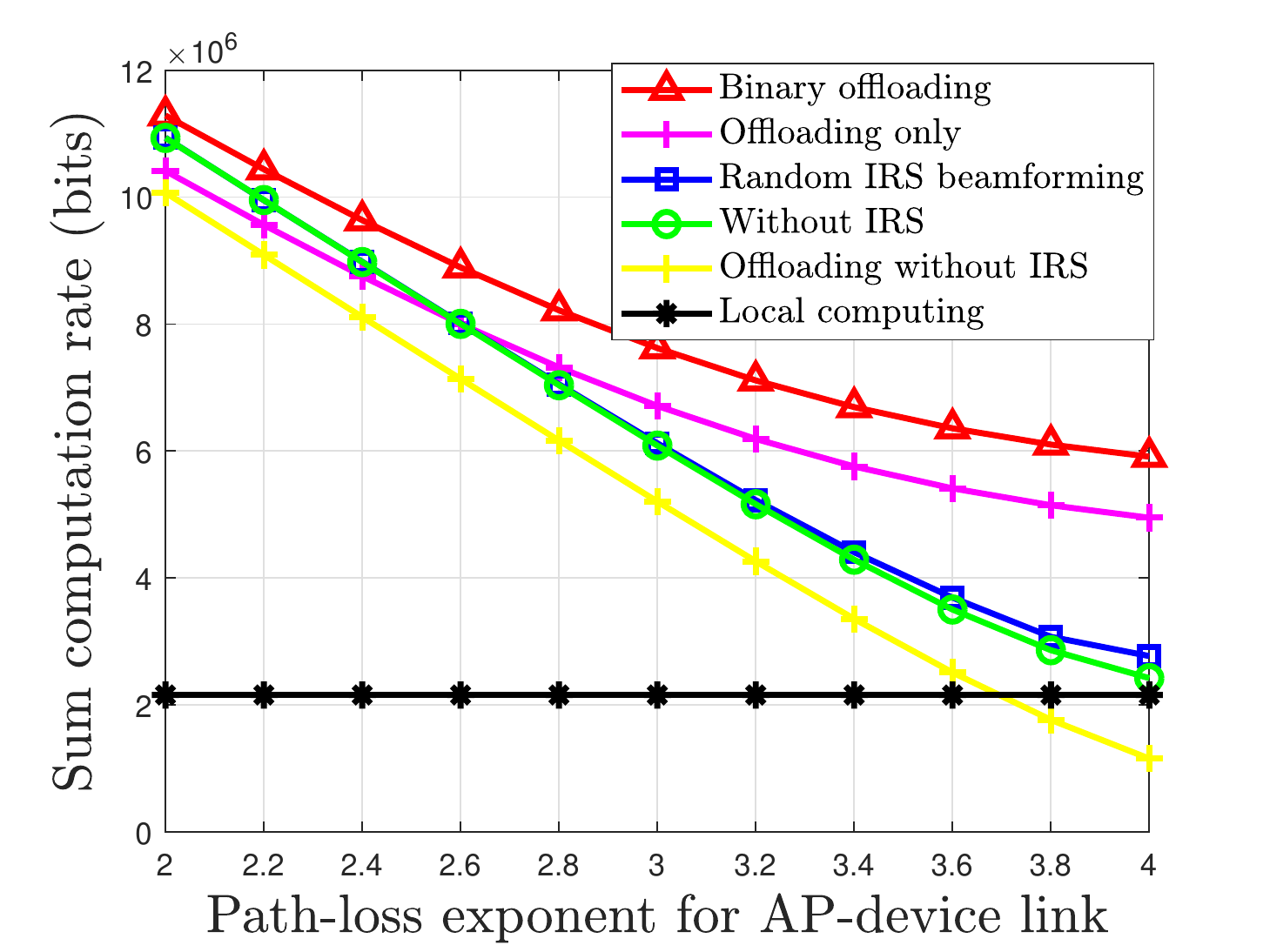}
\caption{Sum computation rate versus the path-loss exponent for the AP-device link with C = 1000 cycles/bit.}
\label{alpha}
\end{minipage}
\vspace{-16pt}
\end{figure}
In Fig. \ref{distance}, we investigate the coverage performance of our considered IRS-aided MEC systems, by plotting the sum computation rate versus the distance between the center of devices cluster and the AP. Note that the IRS also moves accordingly to keep the relative distance with the devices center unchanged. From Fig. \ref{distance}, it is expected that the sum computation rate of all the considered schemes except local computing decreases as the distance increases. Moreover, it is observed that for the schemes without IRS, the computation rate decreases drastically and the gains over local computing vanish rapidly as devices move far away from the AP. This thus fundamentally limits the operating coverage of MEC systems. In contrast, by deploying the IRS in the proximity of devices, the gains of our proposed scheme over local computing can be maintained for a wide range of distance. In other words, for the same available energy at each device, the MEC operating coverage can be extended without compromising the system sum computation rate. For example, for the target sum computational bits about $3 \times {10^6}$, devices with the distance beyond 40 m cannot meet the requirement in the case without IRS, whereas by deploying the IRS in the proximity of devices, the requirement can be still met with the distance of 90 m from the AP. Thus, the offloading coverage increases significantly by deploying the IRS.

\subsubsection{Impact of Path-loss Exponent of the AP-device Link}
In Fig. \ref{alpha}, we plot the system sum computation rate versus the path-loss exponent of the AP-device link. It is observed that the sum computation rate of all schemes except local computing decreases with increasing the path-loss exponent. Moreover, the sum computation rate of MEC systems without IRS is more sensitive to changes of the path-loss exponent, which leads to a faster decrease in the computation rate. By contrast, by deploying the IRS in the proximity of devices, the decrease of the sum computation rate over the path-loss exponent is significantly alleviated. This is mainly because for the systems with IRS, there exists additional DoFs of IRS beamforming optimization for enhancing the reflective link, thereby effectively compensating the performance loss incurred by the higher path-loss exponent. The result demonstrates the effectiveness of deploying IRS for guaranteing the satisfactory offloading efficiency of MEC systems even under the harsh propagation conditions of wireless channels.

\subsubsection{Impact of Available Energy ${E_k}$}
\begin{figure}[t!]
\centering
\includegraphics[width=3in]{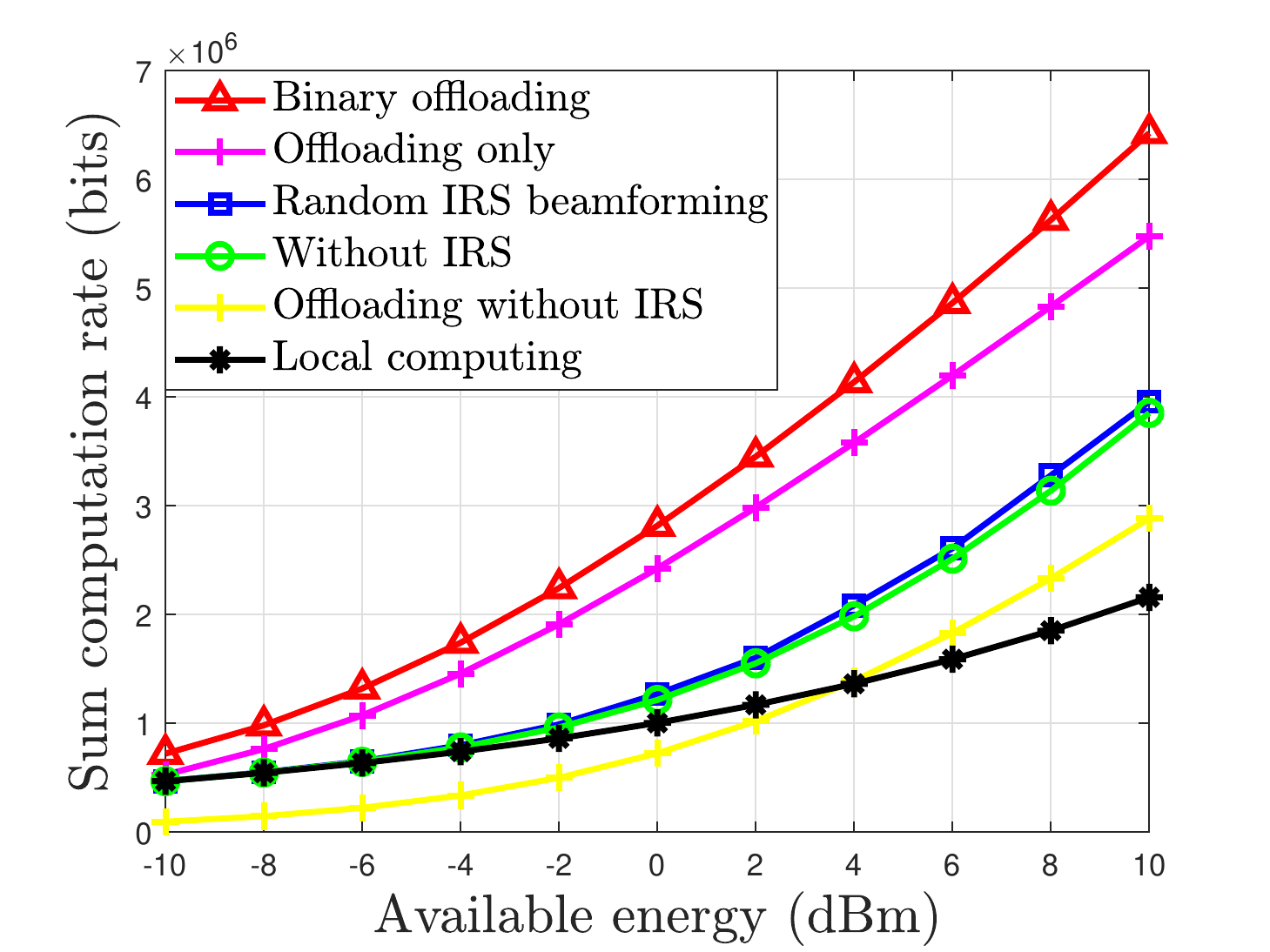}
\caption{{Sum computation rate versus the available energy with C = 1000 cycles/bit.}}
\label{energy}
\vspace{-16pt}
\end{figure}
In Fig. \ref{energy}, we compare the sum computation rate obtained by all schemes versus the available energy at each device. It is observed that our proposed designs can significantly improve the sum computation rate as compared to other benchmarks and the performance gap increases as ${E_k}$ increases, which further demonstrates the benefits of joint optimization of the IRS beamforming and computation mode selection. Moreover, it can be seen that the performance gain of binary  offloading without IRS over local computing is marginal when ${E_k} \le  - 4$ dBm, which indicates that the computational resource at the edge server cannot be fully exploited in the low energy regime of devices. In contrast, by deploying the IRS, the binary offloading scheme significantly outperforms local computing for a wide range of energy available at devices. It implies that the IRS-aided MEC is a promising architecture for supporting multiple energy-limited devices compared to the conventional MEC system.

\section{Conclusion}
This paper developed a unified dynamic IRS framework to maximize the sum computation rate of an MEC system by jointly optimizing the IRS beamforming vectors, computational mode selection, and time allocation. By flexibly controlling the number IRS reconfiguration times, the proposed dynamic IRS beamforming framework is capable of balancing the performance gain and associated signalling overhead. We proposed two algorithms, namely penalty-based SCA algorithm and successive refinement algorithm, to obtain high-quality solutions for the corresponding optimization problem. Simulation results demonstrated the effectiveness of our proposed algorithms, and also illustrated the practical significance of IRS in MEC systems for achieving coverage extension and supporting multiple energy-limited devices for task offloading. Moreover, the maximum number of required
IRS beamforming vectors for achieving the maximum computation rate performance was also revealed, which suggests the fundamental performance-cost tradeoff in exploiting the dynamic IRS beamforming framework.

\section*{Appendix A: \textsc{Proof of Theorem 1}}
For any given the device sets ${{\cal K}_{{\rm{off}}}}$ and ${{\cal K}_{{\rm{loc}}}}$, problem \eqref{C7} is reduced to the following optimization problem
\begin{subequations}\label{C8}
\begin{align}
\label{C8-a}\mathop {\max }\limits_{\left\{ {{\tau _k}} \right\}}  \;\;&B\sum\limits_{k \in {{\cal K}_{{\rm{off}}}}} {{\tau _k}{{\log }_2}} \left( {1 + \frac{{{E_k}{{\left| {{h_{d,k}} + {\bf{q}}_k^H{\bf{v}}_{\Omega \left( k \right)}^*} \right|}^2}}}{{{\tau _k}{\sigma ^2}}}} \right)\\
\label{C8-b}{\rm{s.t.}}\;\;&\sum\nolimits_{k \in {{\cal K}_{{\rm{off}}}}} {{\tau _k} \le T},\\
\label{C8-c}&{\tau _k} \ge 0, ~\forall {k} \in {{\cal K}_{{\rm{off}}}}.
\end{align}
\end{subequations}
It can be readily verified that problem \eqref{C8} is a convex optimization problem since the objective function \eqref{C8-a} is concave and constraints \eqref{C8-b}-\eqref{C8-c} are convex. As such, we exploit the KKT conditions to characterize the optimal value of problem \eqref{C8}. The partial Lagrangian function of problem \eqref{C8} is given by
\begin{align}\label{Lagrangian_function}
{\cal L}\left( {{\tau _k},\mu } \right) = B\sum\limits_{k \in {{\cal K}_{{\rm{off}}}}} {{\tau _k}{{\log }_2}} \left( {1 + \frac{{{E_k}{{\left| {{h_{d,k}} + {\bf{q}}_k^H{\bf{v}}_{\Omega \left( k \right)}^*} \right|}^2}}}{{{\tau _k}{\sigma ^2}}}} \right) + \mu \left( {T - \sum\limits_{k \in {{\cal K}_{{\rm{off}}}}} {{\tau _k}} } \right),
\end{align}
where $\mu  \ge 0$ denotes the dual variable associated with constraint \eqref{C8-b}. According to KKT conditions, we have
\begin{align}\label{KKT_condition}
\frac{{\partial {\cal L}\left( {{\tau _k},\mu } \right)}}{{\partial {\tau _k}}} = \Upsilon \left( {{\gamma _k}} \right) \buildrel \Delta \over = B{\log _2}\left( {1 + {\gamma _k}} \right) - B\frac{{{\gamma _k}}}{{\left( {1 + {\gamma _k}} \right)\ln 2}} - \mu  = 0, ~\forall k \in {{\cal K}_{{\rm{off}}}},
\end{align}
where
\begin{align}\label{gamma_K}
{\gamma _k} = \frac{{{E_k}{{\left| {{h_{d,k}} + {\bf{q}}_k^H{\bf{v}}_{\Omega \left( k \right)}^*} \right|}^2}}}{{{\tau _k}{\sigma ^2}}}, ~\forall k \in {{\cal K}_{{\rm{off}}}},
\end{align}
is the received signal-to-noise ratio (SNR) of offloading device $k$ at the AP. By taking the first order of derivative of $\Upsilon \left( {{\gamma _k}} \right)$ with respect to ${{\gamma _k}}$, it can be verified that $\Upsilon \left( {{\gamma _k}} \right)$ is an increasing function. Since $\Upsilon \left( 0 \right) =  - \mu  \le 0$, the equation $\Upsilon \left( {{\gamma _k}} \right) = 0$ has a unique solution denoted by ${\gamma ^*}$, which can be obtained by bisection search. As such, all the devices operating in offloading mode share the same SNR, i.e,
\begin{align}\label{SNR_relation}
\frac{{{E_k}{{\left| {{h_{d,k}} + {\bf{q}}_k^H{\bf{v}}_{\Omega \left( k \right)}^*} \right|}^2}}}{{\tau _k^*{\sigma ^2}}} = \frac{{{E_j}{{\left| {{h_{d,j}} + {\bf{q}}_j^H{\bf{v}}_{\Omega \left( j \right)}^*} \right|}^2}}}{{\tau _j^*{\sigma ^2}}},\forall k,j \in {{\cal K}_{{\rm{off}}}}.
\end{align}
Based on \eqref{SNR_relation}, the optimal value of problem \eqref{C7} can be written in a closed-form expression as
\begin{align}\label{optimal_value}
R_{{\rm{off}}}^*\left( {{{\cal K}_{{\rm{off}}}}} \right) &= B\sum\limits_{k \in {{\cal K}_{{\rm{off}}}}} {\tau _k^*{{\log }_2}} \left( {1 + \frac{{{E_k}{{\left| {{h_{d,k}} + {\bf{q}}_k^H{\bf{v}}_{\Omega \left( k \right)}^*} \right|}^2}}}{{\tau _k^*{\sigma ^2}}}} \right)\nonumber\\
& = B\sum\limits_{k \in {{\cal K}_{{\rm{off}}}}} {\tau _k^*{{\log }_2}} \left( {1 + \frac{{\sum\nolimits_{k \in {{\cal K}_{{\rm{off}}}}} {{E_k}{{\left| {{h_{d,k}} + {\bf{q}}_k^H{\bf{v}}_{\Omega \left( k \right)}^*} \right|}^2}} }}{{\sum\nolimits_{k \in {{\cal K}_{{\rm{off}}}}} {\tau _k^*{\sigma ^2}} }}} \right)\nonumber\\
&\mathop  = \limits^{\left( a \right)} BT{\log _2}\left( {1 + \frac{{\sum\nolimits_{k \in {{\cal K}_{{\rm{off}}}}} {{E_k}{{\left| {{h_{d,k}} + {\bf{q}}_k^H{\bf{v}}_{\Omega \left( k \right)}^*} \right|}^2}} }}{{T{\sigma ^2}}}} \right),
\end{align}
where (a) follows that $T = \sum\nolimits_{k \in {{\cal K}_{{\rm{off}}}}} {\tau _k^*}$. By adding a device $k' \in {{\cal K}_{{\rm{loc}}}}$ to ${{{\cal K}_{{\rm{off}}}}}$, we have
\begin{align}\label{optimal_value2}
&R_{{\rm{off}}}^*\left( {{{\cal K}_{{\rm{off}}}} \cup \left\{ {k'} \right\}} \right)\nonumber\\
&= BT{\log _2}\left( {1 + \frac{{{E_{k'}}{{\left| {{h_{d,k'}} + {\bf{q}}_{k'}^H{\bf{v}}_{\Omega \left( {k'} \right)}^*} \right|}^2} + \sum\nolimits_{k \in {{\cal K}_{{\rm{off}}}}} {{E_k}{{\left| {{h_{d,k}} + {\bf{q}}_k^H{\bf{v}}_{\Omega \left( k \right)}^*} \right|}^2}} }}{{T{\sigma ^2}}}} \right).
\end{align}
It can be easily verified that when
\begin{align}\label{condition2}
R_{{\rm{off}}}^*\left( {{{\cal K}_{{\rm{off}}}} \cup \left\{ {k'} \right\}} \right) - R_{{\rm{off}}}^*\left( {{{\cal K}_{{\rm{off}}}}} \right) \ge \min \left( {\frac{{T{f_{\max }}}}{{{C_{k'}}}},\frac{{{T^{2/3}}}}{{{C_{k'}}}}{{\left( {\frac{{{E_{k'}}}}{{{\gamma _c}}}} \right)}^{1/3}}} \right),
\end{align}
activating device $k'$ for task offloading can achiever higher computation rate, which thus completes the proof.
\section*{Appendix A: \textsc{Proof of roposition 1}}
Without loss of generality, we suppose that the offloading device set obtained by Algorithm 2 is denoted by ${{\cal K}_{{\rm{off}}}^ \star} = \left\{ {\Phi \left( 1 \right),\Phi \left( 2 \right), \ldots ,\Phi \left( {{K_{{\rm{off}}}}} \right)} \right\}$. When ${E_1} = {E_2} =  \ldots  = {E_K}$ and ${C_1} = {C_2} =  \ldots  = {C_K}$, it follows that ${g_{\Phi \left( 1 \right)}} \ge  \ldots  \ge {g_{\Phi \left( {{K_{{\rm{off}}}}} \right)}} \ge {g_{\Phi \left( {{K_{{\rm{off}}}} + 1} \right)}} \ge  \ldots  \ge {g_{\Phi \left( K \right)}}$. For notational simplicity, we use ${R^*}\left( {{{\cal K}_{{\rm{off}}}}} \right)$ to represent the maximum computation rate based on the set ${{{\cal K}_{{\rm{off}}}}}$. Then, we only need to prove that the computation rate achieved based on ${{\cal K}_{{\rm{off}}}^ \star}$ is higher than that of any other set. Specifically, we consider the following three cases. and all other cases can be directly extended from the study of these three cases.

1) \textit{Case 1}: By adding an element from the set ${\cal K}\backslash {{\cal K}_{{\rm{off}}}^ \star}$, we obtain a new offloading device set denoted by ${{\mathord{\buildrel{\lower3pt\hbox{$\scriptscriptstyle\frown$}}
\over {\cal K}} }_{{\rm{off}}}} = \left\{ {\Phi \left( 1 \right), \ldots ,\Phi \left( {{K_{{\rm{off}}}}} \right),\Phi \left( j \right)} \right\}$, where $\Phi \left( j \right) \in {\cal K}\backslash {{\cal K}_{{\rm{off}}}^ \star}$. It follows that ${g_{\Phi \left( {{K_{{\rm{off}}}}} \right)}} \ge {g_{\Phi \left( {{K_{{\rm{off}}}} + 1} \right)}} \ge {g_{\Phi \left( j \right)}}$. Based on Theorem 1, we have ${R^*}\left( {{{\cal K}_{{\rm{off}}}^ \star} \cup \left\{ {\Phi \left( {{K_{{\rm{off}}}} + 1} \right)} \right\}} \right) \le {R^*}\left( {{{\cal K}_{{\rm{off}}}^ \star}} \right)$. Since ${g_{\Phi \left( {{K_{{\rm{off}}}} + 1} \right)}} \ge {g_{\Phi \left( j \right)}}$, ${R^*}\left( {{{\cal K}_{{\rm{off}}}^ \star} \cup \left\{ {\Phi \left( {{K_{{\rm{off}}}} + 1} \right)} \right\}} \right) \le {R^*}\left( {{{\cal K}_{{\rm{off}}}} \cup \left\{ {\Phi \left( {{K_{{\rm{off}}}} + 1} \right)} \right\}} \right)$. As such, we can obtain ${R^*}\left( {{{\mathord{\buildrel{\lower3pt\hbox{$\scriptscriptstyle\frown$}}
\over {\cal K}} }_{{\rm{off}}}}} \right) \le {R^*}\left( {{{\cal K}_{{\rm{off}}}^ \star}} \right)$, which indicates that adding an arbitrary device from the set ${\cal K}/{{\cal K}_{{\rm{off}}}}$ to ${{\cal K}_{{\rm{off}}}^ \star}$ would result in lower computation rate.

2) \textit{Case 2}: By removing an arbitrary element ${\Phi \left( i \right)}$ from ${{\cal K}_{{\rm{off}}}^ \star}$, we obtain a new offloading device set denoted by ${{\mathord{\buildrel{\lower3pt\hbox{$\scriptscriptstyle\smile$}}
\over {\cal K}} }_{{\rm{off}}}} = \left\{ {\Phi \left( 1 \right), \ldots ,\Phi \left( {{K_{{\rm{off}}}}} \right)} \right\}\backslash \left\{ {\Phi \left( i \right)} \right\}$. Generally, it follows that ${g_{\Phi \left( {{K_{{\rm{off}}}}} \right)}} \le {g_{\Phi \left( i \right)}}$. According to Theorem 1, it can be easily shown that ${R^*}\left( {{{\cal K}_{{\rm{off}}}^ \star}\backslash \left\{ {\Phi \left( {{K_{{\rm{off}}}}} \right)} \right\}} \right) \le {R^*}\left( {{{\cal K}_{{\rm{off}}}^ \star}} \right)$. Since ${g_{\Phi \left( {{K_{{\rm{off}}}}} \right)}} \le {g_{\Phi \left( i \right)}}$, we have ${R^*}\left( {{{\mathord{\buildrel{\lower3pt\hbox{$\scriptscriptstyle\smile$}}
\over {\cal K}} }_{{\rm{off}}}}} \right) \le {R^*}\left( {{{\cal K}_{{\rm{off}}}}\backslash \left\{ {\Phi \left( {{K_{{\rm{off}}}}} \right)} \right\}} \right)$. As such, we can obtain ${R^*}\left( {{{\mathord{\buildrel{\lower3pt\hbox{$\scriptscriptstyle\smile$}}
\over {\cal K}} }_{{\rm{off}}}}} \right) \le {R^*}\left( {{{\cal K}_{{\rm{off}}}^ \star}} \right)$, which indicates that removing an arbitrary device from ${{\cal K}_{{\rm{off}}}^ \star}$ would result in lower computation rate.

3) \textit{Case 3}: By exchanging an element between ${\cal K}\backslash {{\cal K}_{{\rm{off}}}^ \star}$ and ${{\cal K}_{{\rm{off}}}^ \star}$, we obtain a new offloading device set denoted by ${{\tilde {\cal K}}_{{\rm{off}}}} = \left\{ {\Phi \left( 1 \right), \ldots ,\Phi \left( {{K_{{\rm{off}}}}} \right),\Phi \left( j \right)} \right\}\backslash \left\{ {\Phi \left( i \right)} \right\}$, where ${\Phi \left( i \right)} \in {\cal K}_{{\rm{off}}}^ \star$ and $\Phi \left( j \right) \in {\cal K}\backslash {{\cal K}_{{\rm{off}}}^ \star}$. It follows that ${g_{\Phi \left( j \right)}} \le {g_{\Phi \left( i \right)}}$. As such, it can be  easily obtained ${R^*}\left( {{{\tilde {\cal K}}_{{\rm{off}}}}} \right) \le {R^*}\left( {{{\cal K}_{{\rm{off}}}^ \star}} \right)$, which indicates that the maximum computation rate based on ${{\tilde {\cal K}}_{{\rm{off}}}}$ is lower than that based on ${{\cal K}_{{\rm{off}}}^ \star}$.

Note that all other cases can be directly extended from the study of the above three cases. Thus, we complete the proof.

\bibliographystyle{IEEEtran}
\bibliography{IEEEabrv,D:/MEC_binary/myref}


\end{document}